\documentclass[a4paper,UKenglish,cleveref, autoref, thm-restate]{lipics-v2021}

\title{An Empirical Evaluation of \texorpdfstring{$k$}{k}-Means Coresets} 

\author{Chris Schwiegelshohn}{Department of Computer Science, Aarhus University, Denmark }{schwiegelshohn@cs.au.dk}{[orcid]}{Independent Research Fund Denmark (DFF) Sapere Aude Research Leader grant No 1051-00106B.}

\author{Omar Ali Sheikh-Omar
}{Department of Computer Science, Aarhus University, Denmark}{omar@cs.au.dk}{https://orcid.org/0000-0002-0042-5231}{Innovation Fund Denmark under grant agreement No 0153-00233A.}

\authorrunning{C. Schwiegelshohn and O.A. Sheikh-Omar} 

\Copyright{Chris Schwiegelshohn and Omar Ali Sheikh-Omar} 



\category{} 

\relatedversion{} 

\nolinenumbers 

\EventEditors{Shiri Chechik, Gonzalo Navarro, Eva Rotenberg, and Grzegorz Herman}
\EventNoEds{4}
\EventLongTitle{30th Annual European Symposium on Algorithms (ESA 2022)}
\EventShortTitle{ESA 2022}
\EventAcronym{ESA}
\EventYear{2022}
\EventDate{September 5--9, 2022}
\EventLocation{Berlin/Potsdam, Germany}
\EventLogo{}
\SeriesVolume{244}
\ArticleNo{61}

\usepackage{booktabs} 
\usepackage{dsfont} 
\usepackage{longtable}
\usepackage{float}

\newtheorem{fact}[theorem]{Fact}

\newcommand{\dist}{\text{dist}}
\newcommand{\eps}{\varepsilon}
\newcommand{\opt}{\text{OPT}}
\newcommand{\cost}{\text{cost}}
\newcommand{\calS}{\mathcal{S}}

\newcommand{\one}{\mathds{1}}

\begin{document}

\maketitle

\begin{abstract}
Coresets are among the most popular paradigms for summarizing data. In particular, there exist many high performance coresets for clustering problems such as $k$-means in both theory and practice. Curiously, there exists no work on comparing the quality of available $k$-means coresets. 

In this paper we perform such an evaluation. There currently is no algorithm known to measure the distortion of a candidate coreset. We provide some evidence as to why this might be computationally difficult.
To complement this, we propose a benchmark for which we argue that computing coresets is challenging and which also allows us an easy (heuristic) evaluation of coresets. Using this benchmark and real-world data sets, we conduct an exhaustive evaluation of the most commonly used coreset algorithms from theory and practice.
\end{abstract}

\section{Introduction}

The design and analysis of scalable algorithms has become an important research area over the past two decades. This is particularly important in data analysis, where even polynomial running time might not be enough to handle proverbial \emph{big data} sets.
One of the main approaches to deal with the scalability issue is to compress or sketch large data sets into smaller, more manageable ones. The aim of such compression methods is to preserve the properties of the original data, up to some small error, while significantly reducing the number of data points.

Among the most popular and successful paradigms in this line of research are \emph{coresets} \cite{MunteanuS18}. Informally, given a data set $A$, a coreset $\Omega \subset A$ with respect to a given set of queries $Q$ and query function $f: A\times Q \rightarrow \mathbb{R}_{\geq 0}$ approximates the behaviour of $A$ for all queries up to some multiplicative distortion $D$ via $\sup_{q\in Q} \max\left( \frac{f(\Omega,q)}{f(A,q)},\frac{f(A,q)}{f(\Omega,q)}\right) \leq D.$
Coresets have been applied to a number of problems such as computational geometry \cite{AHV05,Chan09}, linear algebra \cite{IndykMGR20,maalouf2019fast}, and machine learning \cite{MMR21,MunteanuSSW18}. But the by far most intensively studied and arguably most successful applications of the coreset framework is the $k$-clustering problem.

Here we are given $n$ points $A$ with (potential unit) weights $w:A\rightarrow \mathbb{R}_{\geq 0}$ in some metric space with distance function $\dist$ and aim to find a set of $k$ centers $C$ such that 
\begin{equation*}
\cost_A(C):= \frac{1}{n} \sum_{p\in A}  \min_{c\in C} w(p)\cdot \dist^z(p,c)
\end{equation*}
is minimized. The most popular variant of this problem is probably the $k$-means problem in $d$-dimensional Euclidean space where $z=2$ and $\dist(x,y) = \sqrt{\sum_{i=1}^d (x_i-y_i)^2}$.

A $(k,\varepsilon)$-coreset is now a subset $\Omega\subset A$ with weights $w:\Omega\rightarrow \mathbb{R}_{\geq 0}$ such that for any set of $k$ centers $C$
\begin{equation}
\label{eq:coreset}
\sup_{C} \max\left( \frac{\cost_A(C)}{\cost_{\Omega}(C)},\frac{\cost_{\Omega}(C)}{\cost_{A}(C)}\right) \leq 1+\varepsilon.
\end{equation}
The coreset definition in \cref{eq:coreset} provides an upper bound for the distortion of all candidate solutions i.e., all possible sets of $k$ centers. 
A \emph{weak coreset} is a relaxed guarantee that holds for optimal or nearly optimal clusterings of $A$ instead of all clusterings.

In a long line of work spanning the last 20 years
\cite{BecchettiBC0S19,BravermanJKW21,Chen09,Cohen-AddadSS21b,Cohen-AddadSS21,FeldmanL11,FeldmanSS20,
HaM04,HaK07,HuangJLW18,huang2020coresets,BravermanJKW21,
LangbergS10,SohlerW18}, the size of coresets has been steadily improved with the current state of the art yielding a coreset with $\tilde{O}(k \varepsilon^{-2} \cdot \min(d,k,\varepsilon^{-2}))$ points for a distortion $D\leq (1+\varepsilon)$ due to \cite{CLSS22}\footnote{We use $\tilde O(x)$ to hide $\log^c x$ terms for any constant $c$.}.

While we have a good grasp of the theoretical guarantees of these algorithms, our understanding of the empirical performance is somewhat lacking. There exist a number of coreset implementations, but it is usually difficult to assess which implementation summarizes the data best. To accurately evaluate a given coreset, we would need to come up with a $k$ clustering $C$ which results in a maximal distortion. Solving this problem is likely difficult: related questions such as deciding whether a 3-dimensional point set $A$ is an $\varepsilon$-net of a set $B$ with respect to convex ranges is co-NP hard \cite{GiannopoulosKWW12}. 
It is similarly co-NP hard to decide whether a point set $A$ is a \emph{weak coreset} of a point set $B$ (see \cref{prop:hardness} in the appendix). 

Due to this difficulty, a common heuristic for evaluating coresets is as follows~\cite{AckermannMRSLS12,FGSSS13}. First, compute a coreset $\Omega$ with the available algorithm(s) using some input data $A$. Then, run an optimization algorithm on $\Omega$ to compute a $k$ clustering. The \emph{best} coreset algorithm is considered to be the one which yields a clustering with the smallest cost.

This practice has substantial drawbacks.
The first is that this evaluation method conflates the two separate tasks of coreset construction and optimization.
It is important to note that the first step of virtually all coreset algorithms is a low-cost (bicriteria) constant factor approximation, i.e. a solution with $\beta\cdot k$ clusters that costs at most $\alpha\cdot \opt$, where $\opt$ is the cost of an optimal $k$ clustering.
Given that this initial solution has an $\alpha$ approximation to the cost, a routine calculation shows that the additive error of the coreset, i.e. the maximum difference
$ \left\vert \cost_A(C) - \cost_B(C) \right\vert $
over all solutions $C$ is at most $O(\alpha)\cdot \cost_A(C)$. 
In particular, in the case that the initial bicriteria approximation has $\alpha \ll 2$, which is not too difficult to achieve with more than $k$ centers, any $\gamma$ approximation algorithm will find solutions with approximation factor $O(\gamma + \alpha) \cdot\opt$. In particular, the distortion may be unbounded, for example if $B$ only consists of the $k$ centers, while simply returning $B$ itself yields a low cost clustering. Thus, it is difficult to measure coreset quality in this way.

The second drawback is that this practice will mainly measure the performance of the optimization algorithm, rather than the performance of the coreset algorithms. During its execution it might simply not consider any solution with high distortion. For example, if the approximation factor $\gamma$ of the solution returned by the algorithm is large then this solution (as well as any even higher cost solution considered during the algorithm’s execution) will have a low distortion. 

The third drawback of this evaluation method is that it does not consider the main use cases of coresets, nor the full power of their guarantee. Indeed, if speeding up the computation of an optimization algorithm, one would hardly need a strong coreset; approximating the cost of every candidate solution, as weaker coreset definitions (or indeed a bicriteria approximation) would be suitable as well. A coreset's main and most powerful feature is \emph{composability}, i.e. given two disjoint point sets $X$ and $Y$, the union of a coreset of $X$ and a coreset of $Y$ is a coreset. Composability is what enables coresets to scale to massively parallel computation models and enables simple streaming algorithms via the merge and reduce technique. To which degree a coreset is composable is generally not a property of an optimal clustering of the point set, as optimal solutions $C_X$ of $X$ or $C_Y$ of $Y$ may have little in common with an optimal solution of $X\cup Y$.

The purpose of this study is to systematically evaluate the quality of various coreset algorithms for $k$-means. As such, we develop a new evaluation procedure which estimates the distortion of coreset algorithms. On real-world data sets, we observe that while the evaluated coreset algorithms are generally able to find solutions with comparable costs, there is a stark difference in their distortions. This shows that differences between optimization and compression are readily observable in practice.

As a complement to our evaluation procedure on real-world data sets, we propose a benchmark framework for generating synthetic data sets. We argue why this benchmark has properties that results in hard instances for all known coreset constructions. We also show how to efficiently estimate the distortion of a candidate coreset on the benchmark.

\section{Coreset Algorithms}
\label{sec:algorithms}

Though the algorithms vary in details, coreset constructions come in one of the following two flavours:

\begin{enumerate}
\item {\bf Movement-based constructions:} Such algorithms compute a coreset $\Omega$ with $T$ points given some input point set $A$ such that $\cost_{\Omega}(C)\ll \opt$, where $\opt$ is the cost of an optimal $k$-means clustering of $A$. 
The coreset guarantee then follows as a consequence of the triangle inequality. These algorithms all have an exponential dependency on the dimension $d$, and therefore have been overtaken by sampling-based methods. Nevertheless, these constructions are more robust to various constrained clustering formulations~\cite{HuangJV19,SSS19} and continue to be popular. Examples from theory include~\cite{FrahlS2005,HaM04}. 

\item {\bf Importance sampling:} Points are sampled proportionate to their impact on the cost of any given candidate solution. The idealized distribution samples proportionate to the sensitivity which for a point $p$ is defined as $sens(p):=\sup_{C} \frac{\min_{c\in C} \dist^2(p,c)}{\cost_A(C)}$ and weighted by their inverse sampling probability. The sensitivities are hard to compute exactly but much work exists on how to find other distributions with very similar properties. In terms of theoretical performance, sensitivity sampling has largely replaced movement-based constructions, see for example~\cite{FeldmanL11,LangbergS10}.  
\end{enumerate}

Of course, there exist algorithms that draw on techniques from both, see for example~\cite{Cohen-AddadSS21}. In what follows, we will survey implementations of various coreset constructions that we will evaluate later.

{\bf StreamKM++~\cite{AckermannMRSLS12}:} The popular $k$-means++ algorithm~\cite{ArV07} computes a set of centers $K$ by iteratively sampling a point $p$ in $A$ proportionate to $\min_{q\in K} \dist^2(p,q)$ and adding it to $K$. The procedure terminates once the desired number of centers has been reached. The first center is typically picked uniformly at random.
The StreamKM++ paper runs the $k$-means++ algorithms for $T$ iterations, where $T$ is the desired coreset size. At the end, every point $q$ in $K$ is weighted by the number of points in $A$ closest to it. While the construction has elements of importance sampling, the analysis is largely movement-based. The provable bound required for the algorithm to compute a coreset is $O\left(\frac{k\log n}{\delta^{d/2}\varepsilon^d}\cdot \log^{d/2} \frac{k\log n}{\delta^{d/2}\varepsilon^d}\right)$. Despite its simplicity, its running time compares unfavourably to all other constructions.

{\bf BICO~\cite{FGSSS13}:} BICO combines the very fast, but poor quality clustering algorithm BIRCH~\cite{ZRL97} with the movement-based analysis from~\cite{FrahlS2005,HaM04}. The clustering is organized by way of a hierarchical decomposition: When adding a point $p$ to one of the coreset points $\Omega$ at level $i$, it first finds the closest point $q$ in $\Omega$. If $p$ is too far away from $q$, a new cluster is opened with center at $p$. Otherwise $p$ is either added to the same cluster as $q$, or, if adding $p$ to $q$'s cluster increases the clustering cost beyond a certain threshold, the algorithm attempts to add $p$ to the child-clusters of $q$. The procedure then continues recursively. The provable bound required for the algorithm to compute a coreset is $O\left(k\varepsilon^{-d-2}\log n\right)$.

{\bf Ray Maker~\cite{HaK07}:} The algorithm computes an initial solution with $k$ centers which is a constant factor approximation of the optimal clustering. Around each center, $O(1/\epsilon^{d-1})$ random rays are created which span the hyperplane. Next, each point $p \in A$ is snapped to its closest ray resulting in a set of one-dimensional points associated with each ray. Afterwards, a coreset is created for each ray by computing an optimal 1D clustering with $k^2/\epsilon^2$ centers and weighing each center by the number of points in each cluster. The final coreset is composed of the coresets computed for all the rays.
The provable bound required for the algorithm to compute a coreset is $O(k^3 \cdot \varepsilon^{-d-1})$. The algorithm has recently received some attention due to its applicability to the fair clustering problem~\cite{HuangJV19}.

{\bf Sensitivity Sampling~\cite{FeldmanL11}:} The simplest implementation of sensitivity sampling first computes an $(O(1),O(1))$ bicriteria approximation\footnote{An $(\alpha,\beta)$ bicriteria approximation computes an $\alpha$ approximation using $\beta\cdot k$ many centers.}, for example by running $k$-means++ for $2k$ iterations~\cite{Wei16}. Let $K$ be the $2k$ clustering thus computed and let $K_i$ be an arbitrary cluster of $K$ with center $q_i$. Subsequently, the algorithm picks points proportionate to $\frac{\dist^2(p,q)}{\cost_{K_i}(\{q_i\})} + \frac{1}{|K_i|}$ and weighs any point by its inverse sampling probability. Let $|\hat{K_i}|$ be the estimated number of points in the sample. Finally, the algorithm weighs each $q_i$ by $(1+\eps)\cdot |K_i| - |\hat{K_i}|$. The provable bound required for the algorithm to compute a coreset is $\tilde O\left(kd\varepsilon^{-4}\right)$ (\cite{FeldmanL11}),
$\tilde O\left(k\varepsilon^{-6}\right)$ (\cite{huang2020coresets}), or $\tilde O\left(k^2\varepsilon^{-4}\right)$ (\cite{BravermanJKW21}).

{\bf Group Sampling~\cite{Cohen-AddadSS21}:} First, the algorithm computes an $O(1)$ approximation (or a bicriteria approximation) $K$. Subsequently, the algorithm preprocesses the input into groups such that (1) for any two points $p,p'\in K_i$, their cost is identical up to constant factors and (2) for any two clusters $K_i,K_j$, their cost is identical up to constant factors. In every group, Group Sampling now samples points proportionate to their cost. The authors of~\cite{Cohen-AddadSS21} show that there always exist a partitioning into $\log^2 1/\varepsilon$ groups. Points not contained in a group are snapped to their closest center $q$ in $K$. $q$ is weighted by the number of points snapped to it. The provable bound required for the algorithm to compute a coreset is $\tilde O\left(k\varepsilon^{-2}\min(d,k,\varepsilon^{-2})\right)$ (\cite{CLSS22}). While this improves over sensitivity sampling, it is generally slower and not as easy to implement.

Finally, we note that some of the more popular algorithms in theory have not been mentioned here. For example, Chen's \cite{Chen09} construction is particularly popular among theoreticians. The Group Sampling algorithm by \cite{Cohen-AddadSS21} is an extension and improvement of Chen's method. Thus, the performance of Group Sampling is also indicative of Chen's algorithm.

\paragraph*{Dimension Reduction}
Finally, we also combine coreset constructions with a variety of dimension reduction techniques. Starting with~\cite{DrineasFKVV04}, a series of results \cite{BecchettiBC0S19,BoutsidisMD09,BoutsidisZD10,BoutsidisZMD15,CEMMP15,Cohen-AddadS17,FeldmanSS20,FKW19,KuK10,MakarychevMR19,SohlerW18} explored the possibility of using dimension reduction methods for $k$-clustering, with a particular focus on principal component analysis (PCA) and random projections. The seminal paper by Feldman, Schmidt, and Sohler~\cite{FeldmanSS20} was the first to use dimension reduction to obtain smaller coresets for $k$-means. Movement-based coresets in particular often have an exponential dependency on the dimension, which can be alleviated with some form of dimension reduction, both in theory~\cite{SSS19} and in practice~\cite{KappmeierS015}.
There are essentially two main dimension reduction techniques for coresets.

{\bf Principal Component Analysis:} Feldman, Schmidt, and Sohler~\cite{FeldmanSS20} showed that projecting an input $A$ onto the first $O(k/\varepsilon^2)$ principal components is a coreset. This coreset still consists of $n$ points, but they now lie in low dimension. The analysis was subsequently tightened by~\cite{CEMMP15} and extended to other center-based cost functions by~\cite{SohlerW18}. Although its target dimension is generally worse than those based on random projections and terminal embeddings, there is nevertheless reasons for using PCA regardless: It removes noise and thus may make it easier to compute a high quality coreset. For more applications of PCA to $k$-means clustering, we refer to

{\bf Terminal Embeddings:} Given a set of points $A$ in $\mathbb{R}^D$, a terminal embedding $f:\mathbb{R}^D\rightarrow \mathbb{R}^d$ preserves the pairwise distance between any point $p\in A$ and any point $q\in \mathbb{R}^D$ up to a $(1\pm \varepsilon)$ factor. The statement is related to the famous Johnson-Lindenstrauss lemma but it is stronger as it does not apply to only the pairwise distances of $A$. Nevertheless, the same target dimension is sufficient. Terminal embeddings were studied by~\cite{CherapanamjeriN21,ElkinFN17,MahabadiMMR18,NaN18}, with Narayanan and Nelson \cite{NaN18} achieving an optimal target dimension of $O(\varepsilon^{-2}\log n)$, where $n$ is the number of points. We note that terminal embeddings, combined with an iterative application of the coreset construction from \cite{BravermanJKW21}, can reduce the target dimension to a factor $\tilde{O}(\varepsilon^{-2} \log k)$. This is mainly of theoretical interest, as in practice the deciding factor wrt the target dimension is the precision, rather than dependencies on $\log n$ and $\log k$. For applications to coresets, we refer to \cite{BecchettiBC0S19,Cohen-AddadSS21,huang2020coresets}. For an empirical evaluation of random projections, which form the basis of all known terminal embeddings, we refer to Venkatsubramanian and Wang~\cite{VenkatasubramanianW11}.

\section{Benchmark Construction}
\label{sec:benchmark}

In this section, we describe our benchmark. We start by describing the aims of the benchmark, followed by giving the construction. Our aim is to generate a data set containing many clusterings with the following properties. 
\begin{enumerate}
\item The benchmark has many clusterings that, in a well defined sense, are highly dissimilar. Specifically, we want the overlap between any two clusters of different clusterings to be small.
\item The different clusterings have very similar and low cost. This ensures that despite the solutions being different in terms of composition and center placement, a good coreset has to consider them equally regarding distortion.
\item The clusterings are induced by a minimal cost assignment of input points to a set of centers in $\mathbb{R}^d$. This final property ensures that the coreset guarantee has to apply to these clusterings.
\end{enumerate}

To generate the benchmark, we now use the following construction. The benchmark has a parameter $\alpha$ which controls the number of points and dimensions of the generated data instance.
For a given value of $k$, the benchmark instance consists of $n=k^\alpha$ points and $d=\alpha \cdot k$ dimensions, i.e. we will construct and $n\times d$ matrix $A$ where every row corresponds to an input point and every column corresponds to one of the dimensions.

Let $\one_k$ be the $k$-dimensional all-one vector and $v_i^1$ be the $k$-dimensional vector with entries $(v_i^1)_j = \begin{cases}-\frac{1}{k} & \text{if } i\neq j\\
\frac{k-1}{k} & \text{if } i= j\end{cases}$.
For $\ell\leq \alpha$, recursively define the $k^{\ell}$ dimensional vector $v_i^{\ell} = v_i^{\ell-1}\otimes \one_k $, where $\otimes$ denotes the Kronecker product, i.e. $ v_i^{\ell-1}\otimes \one_k= \begin{bmatrix}
(v_i^{\ell-1})_1 \cdot \one_k \\
(v_i^{\ell-1})_2 \cdot \one_k \\
\vdots \\
(v_i^{\ell-1})_1 \cdot \one_k
\end{bmatrix}$. 
Finally, set the $t$-th column of $A$, for $t = a\cdot k + b$, $a\in \{0,\ldots \alpha-1\}$ and $b \in \{1,\ldots k\}$, to be $\one_{k^{\alpha-a+1}}\otimes v_b^{a+1}$.

To get a better feel for the construction, we have given two small example instances for $k=2$ and $k=3$ in Figue \cref{fig:benchmark-small-instances}.
\begin{figure*}[h]
\begin{center}
$$ 
\begin{bmatrix}
\frac{1}{2} & -\frac{1}{2} & \frac{1}{2} & -\frac{1}{2} & \frac{1}{2} & -\frac{1}{2}  \\
-\frac{1}{2} & \frac{1}{2} & \frac{1}{2} & -\frac{1}{2} & \frac{1}{2} & -\frac{1}{2} \\
\frac{1}{2} & -\frac{1}{2} & -\frac{1}{2} & \frac{1}{2} & \frac{1}{2} & -\frac{1}{2}  \\
-\frac{1}{2} & \frac{1}{2} & -\frac{1}{2} & \frac{1}{2} & \frac{1}{2} & -\frac{1}{2} \\
\frac{1}{2} & -\frac{1}{2} & \frac{1}{2} & -\frac{1}{2} & -\frac{1}{2} & \frac{1}{2}  \\
-\frac{1}{2} & \frac{1}{2} & \frac{1}{2} & -\frac{1}{2} & -\frac{1}{2} & \frac{1}{2}  \\
\frac{1}{2} & -\frac{1}{2} & -\frac{1}{2} & \frac{1}{2} &-\frac{1}{2} & \frac{1}{2} \\
-\frac{1}{2} & \frac{1}{2} & -\frac{1}{2} & \frac{1}{2}  & -\frac{1}{2} & \frac{1}{2} \\
\end{bmatrix} ~~~~~~~~
\begin{bmatrix}
 \frac{2}{3}   & -\frac{1}{3}  & -\frac{1}{3} & \frac{2}{3}   & -\frac{1}{3}  & -\frac{1}{3}   \\
 -\frac{1}{3}  & \frac{2}{3}   & -\frac{1}{3} & \frac{2}{3}   & -\frac{1}{3}  & -\frac{1}{3}   \\
 -\frac{1}{3}  & -\frac{1}{3}  & \frac{2}{3}  & \frac{2}{3}   & -\frac{1}{3}  & -\frac{1}{3}   \\
 \frac{2}{3}   & -\frac{1}{3}  & -\frac{1}{3} & -\frac{1}{3}  & \frac{2}{3}   & -\frac{1}{3}   \\
 -\frac{1}{3}  & \frac{2}{3}   & -\frac{1}{3} & -\frac{1}{3}  & \frac{2}{3}   & -\frac{1}{3}   \\
 -\frac{1}{3}  & -\frac{1}{3}  & \frac{2}{3}  & -\frac{1}{3}  & \frac{2}{3}   &  -\frac{1}{3}  \\
 \frac{2}{3}   & -\frac{1}{3}  & -\frac{1}{3} & -\frac{1}{3}  & -\frac{1}{3}  & \frac{2}{3}    \\
 -\frac{1}{3}  & \frac{2}{3}   & -\frac{1}{3} & -\frac{1}{3}  & -\frac{1}{3}  & \frac{2}{3}    \\
 -\frac{1}{3}  & -\frac{1}{3}  & \frac{2}{3}  & -\frac{1}{3}  & -\frac{1}{3}  & \frac{2}{3}    \\
\end{bmatrix} 
$$
\end{center}
\caption{Benchmark construction for $k=2$ and $\alpha=3$ (left) and $k=3$ and $\alpha=2$ (right).}
\label{fig:benchmark-small-instances}
\end{figure*}

\paragraph*{Properties of the Benchmark}

We now summarize the key properties of the benchmark.
To this end, we require a few notions.
Let $A$ be the input matrix. We slightly abuse notation and refer to $A_i$ as both the $i$th point as well as the $i$th row of the matrix $A$.
For a clustering $\mathcal{C}=\{C_1,\ldots ,C_k\}$, we define that the $n\times k$ indicator matrix $\tilde X$ induced by $\mathcal{C}$ via $\tilde X_{i,j} = \begin{cases}1 & \text{if } A_i\in C_j \\
0 & \text{else.} \end{cases}$
Furthermore, we will also use the $n\times k$ normalized clustering matrix $ X$ defined as
$X_{i,j} = \begin{cases}\frac{1}{\sqrt{|C_i|}} & \text{if } A_i\in C_j \\
0 & \text{else.} \end{cases}$
We also recall the following lemma which will allow us to express the $k$-means cost of a clustering $\mathcal{C}$ with optimally chosen centers in terms of the cost of $X$ and $A$.
\begin{lemma}[Folklore]
\label{lem:magic}
Let $A$ be an arbitrary set of points and let $\mu(A) = \frac{1}{|A|}\sum_{p\in A} p$ be the mean. Then $ \sum_{p\in A} \|p-c\|^2 = |A|\cdot \|\mu(A)-c\|^2 + \sum_{p\in A} \|p-\mu(A)\|^2$ for any point $c$.
\end{lemma}
This lemma proves that for any given cluster $C_j$, the mean is the optimal choice of center. 
We also note that any two distinct columns of $X$ are orthogonal. Furthermore $\frac{1}{n}\mathbf{1}\mathbf{1}^TA$ copies the mean into every entry of $A$. Combining these two observations, we see that the matrix $XX^TA$ maps the $i$th row of $A$ onto the mean of the cluster it is assigned to. Finally, define the Frobenius norm of an $n\times d$ $A$ by $\|A\|_F = \sqrt{\sum_{i=1}^n\sum_{j=1}^d A_{i,j}^2}$. Then the $k$-means cost of the clustering $\mathcal{C}$ is precisely
$\|A-XX^TA\|_F^2.$

We also require the following distance measure on clusterings as proposed by Meila~\cite{Meila05,Meila06}. Given two clusterings $\mathcal{C}$ and $\mathcal{C'}$, the $k\times k$ confusion matrix $M$ is defined as
$ M_{i,j} = |C_i\cap C'_j|.$
Furthermore for the indicator matrices $\tilde X$ and $\tilde X'$ induced by $\mathcal{C}$ and $\mathcal{C'}$ we have the identity $M=\tilde X^T {\tilde X'}$.
Denote by $\Pi_k$ the set of all permutations over $k$ elements. Then the distance between  $\mathcal{C}$ and $\mathcal{C'}$ is defined as $d(\mathcal{C},\mathcal{C'}) = 1-\frac{1}{n}\underset{\pi\in \Pi_k}{\max} \sum_{i=1}^k M_{i,\pi(i)}.$
Observe that for clusters that are identical, their distance is $0$. The maximum distance between any two $k$ clusterings is always $\frac{k-1}{k}$.

The solutions we consider are given as follows. For the columns $a\cdot k+1,\ldots (a+1)\cdot k$, we define the clustering $\mathcal{C}^{a} = \{C_1^a,\ldots C_k^a\}$ with 
$A_i\in C_j^a$ if and only if $A_{i,j} > 0$. Let $\tilde X^a$ and $X^{a}$ denote the indicator matrix and clustering matrix, respectively, as induced by $\mathcal{C}^{a}$.
These clusterings satisfy the properties we stated at the beginning of this section, that is:
\begin{enumerate}
\item The distance between these clustering is $1-\frac{1}{k}$, i.e. it is maximized.
\item The clusterings have equal cost and the centers in each clustering have equal cost.
\item The clusterings are induced by a set of centers in $\mathbb{R}^d$.
\end{enumerate}

\paragraph*{Benchmark Evaluation}

We now describe how we use the benchmark to measure the distortion of a coreset. Assume for now that the coresets are subsets of the original input points. The extension to coresets that do not consist of input points is described at the end of this section.

Consider the clustering $\mathcal{C}^{a} = \{C_1^a,\ldots C_k^a\}$ for some $a$ and let $\Omega$ with weights $w:\Omega\rightarrow \mathbb{R}_{\geq 0}$ be the coreset and let $\delta>0$ be a parameter. 
Note that there are $\alpha$ many such clusterings, for each value of $a$.
We use $w(C_i^a \cap \Omega):=\sum_{p\in C_i^a \cap \Omega} w(p)$ to denote the mass of points of $C_i^a$ in $\Omega$.
For every cluster $C_i^a$ with $w(C_i^a \cap \Omega)\geq |C_i| (1-\delta)$, we place a center at $\mu(C_i^a)$. Conversely, if $w(C_i^a \cap \Omega)< |C_i^a| (1-\delta)$, we do not place a center at $\mu(C_i^a)$. We call such clusters \emph{deficient}. Let $\calS$ be the centers of these deficient clusters.

We now compare the cost as computed on the coreset and the true cost of $\calS$. Due to \cref{lem:magic} and the fact that all clusters have equal cost, we may write for any deficient cluster $C_i^a$
$\cost_{C_i^a}(\calS) = \cost_{C_j^a}(\{\mu(C_j^a)\}) + k^{\alpha-1}\|\mu(C_j^a)-\mu(C_h^a)\|_2^2$, where $C_h^a$ is a non-deficient cluster.
Thus, the cost is $\cost_{C_i^a}(\calS) \approx \left(1+\frac{2}{\alpha}\right)\cdot \cost_{C_j^a}(\{\mu(C_j^a)\}).$

Conversely, the cost on the coreset is 
\begin{align*}
\cost_{\Omega\cap C_i^a}(\calS)  
 \approx  \frac{w(C_i^a \cap \Omega)}{ \cost_{C_j^a}(\{\mu(C_j^a)\})}\left(1+\frac{2}{\alpha}\right)\cdot \cost_{C_j^a}(\{\mu(C_j^a)\}).
\end{align*}
Thus for each deficient clustering individually, the distortion will be close to $\frac{k^{\alpha-1}}{w(C_i^a \cap \Omega)} \geq \frac{1}{1-\varepsilon}$.
If there are many deficient clusters, then this will also be the overall distortion.
For all possible (suitably discretized) thresholds for deficiency, i.e. all values of $\delta$, we can now identify the clustering $\mathcal{C}^a$ with a maximum number of deficient clusters and use the aforementioned construction to get a lower bound on the distortion.

To extend this evaluation to coresets where the points are not part of the input, we consider a point $p\in \Omega$ to be in $C_i^a$ if it is closer to $\mu(C_i^a)$ than to to $\mu(C_j^a)$.

\section{Experiments} \label{sec:experiments}
In this section, we present how we evaluated different algorithms. First, we propose our evaluation procedure which gauges the quality of coresets. Then, we describe the data sets used for the empirical evaluation and our experimental setup. Finally, we detail the outcome of the experiments and our interpretation of the results.

\subsubsection*{Evaluation Procedure}
\label{sec:evaluation-procedure}
Accurately evaluating a $k$-means coreset of a real-world data set requires constructing a solution (a set of $k$ centers) which results in a maximal distortion. Finding such a solution, however, is difficult. Instead, we can estimate the quality of a given coreset by finding meaningful candidate solutions. 

A first attempt can be to randomly generate candidate solutions. It is not readily apparent how to define a distribution of meaningful solutions from which to sample. One could, for instance, generate $k$ random points inside the convex hull or the minimum enclosing ball (MEB) of a coreset $\Omega$. Convex hulls in high dimensions are infeasible to compute, so we sample a center by choosing random convex combination of the centers of the initial bicriteria approximation computed for every coreset. 
A better way to generate candidate solutions turns out to be $k$-means++, where we sample $k$ points with respect to the $k$-means++ distribution and use the resulting centers as a solution (see \cref{sec:candidate-solution-generation} in the appendix). The main advantage of this approach is that $k$-means++ can uncover natural cluster structures in the data, which uniform sampling generally does not.
For all variants, we generated $5$ candidate solutions, where the candidate solution with the largest distortion being a lower bound for the true distortion of the coreset.

Given the usefulness of evaluating coresets on real-world data sets, it can be tricky to gauge the general performance of coreset algorithms using only a small selection of data sets. For this reason, we used our benchmark to complement the evaluation on real-world data sets. The benchmark accomplishes two important tasks. First, the benchmark allows us to quickly find a bad solution because both good and bad clusterings are known a priori. It is unclear how to find bad clusterings for real-world data sets. Second, it is easier to make a fair comparison of different coreset constructions because the benchmark is known to generate hard instances for all known coreset algorithms. This cannot be said for real-world data sets. For the benchmark, we computed the distortion following the evaluation procedure described in~\cref{sec:benchmark}. 

Every randomized coreset construction was repeated $10$ times. We aggregated the reported maximum distortions for every run by taking the average over all $10$ evaluations. 
It is important to not aggregate the distortions here by taking the maximum over all runs: If one run of the coreset algorithm fails but the others succeed, then such an aggregation predicts far worse distortion than what we could typically expect.

\begin{table}
	\begin{center}
	\begin{tabular}{lrr}
		\toprule
        
		    & Data points
		    & Dimensions
            \\
		\midrule
		\textit{Caltech}
    		& 3,680,458
    		& 128
    		\\
		\textit{Census}
    		& 2,458,285
    		& 68
    		\\
	    \textit{Covertype}
    	    & 581,012
    		& 54
    		\\
	    \textit{NYTimes}
    	    & 500,000
    		& 102,660
    		\\
        \textit{Tower}
            & 4,915,200
    		& 3
    		\\
		\bottomrule
	\end{tabular}\\
	\caption{The sizes of the real-world datasets used for the experimental evaluation}
	\label{tab:real-world-datasets-overview}
	\end{center}
\end{table}

\begin{table}
	\begin{center}
	\begin{tabular}{rrrr}
		\toprule
        $k$
		    & $\alpha$
		    & Data points
		    & Dimensions
            \\
		\midrule
        10
    		& 6
    		& 1,000,000
    		& 60
    		\\
        20
    		& 5
    		& 3,200,000
    		& 100
    		\\
        30
    		& 4
    		& 810,000
    		& 120
    		\\
        40
    		& 4
    		& 2,560,000
    		& 160
    		\\
		\bottomrule
	\end{tabular}\\
	\end{center}
    \caption{The parameter values and the sizes of the benchmark instances used for the experimental evaluation.}
	\label{tab:benchmark-instances-overview}
\end{table}

\subsubsection*{Data sets}
We conducted experiments on five real-world data sets
\textit{Census},
\textit{Covertype},
\textit{Tower},
\textit{Caltech},
\textit{NYTimes},
and four instances of our benchmark. Benchmark instances were generated to match approximately the sizes of the real-world data sets. 
The sizes of the considered data sets are given in Table \ref{tab:real-world-datasets-overview}.

The \textit{Census}\footnote{\url{https://archive.ics.uci.edu/ml/datasets/US+Census+Data+(1990)}} dataset is a small subset of the Public Use Microdata Samples from 1990 US census. It consists of demographic information encoded as 68 categorical attributes of 2,458,285 individuals. 

\textit{Covertype}\footnote{\url{https://archive.ics.uci.edu/ml/datasets/covertype}} is comprised of cartographic descriptions and forest cover type of four wilderness areas in the Roosevelt National Forest of Northern Colorado in the US. It consists of 581,012 records, 54 cartographic variables and one class variable. Although \textit{Covertype} was originally made for classification tasks, it is often used for clustering tasks by removing the class variable~\cite{AckermannMRSLS12}.

The data set with the fewest number of dimensions is \textit{Tower}\footnote{\url{http://homepages.uni-paderborn.de/frahling/coremeans.html}}. This data set consists of 4,915,200 rows and 3 features as it is a 2,560 by 1,920 picture of a tower on a hill where each pixel is represented by a RGB color value. 

Inspired by~\cite{FGSSS13}, \textit{Caltech} was created by computing SIFT features from the images in the Caltech101\footnote{\url{http://www.vision.caltech.edu/Image_Datasets/Caltech101/}} image database. This database contains pictures of objects partitioned into 101 categories. Disregarding the categories, we concatenated the 128-dimensional SIFT vectors from each image into one large data matrix with 3,680,458 rows and 128 columns. 

\textit{NYTimes}\footnote{\url{https://archive.ics.uci.edu/ml/datasets/bag+of+words}} is a dataset composed of the bag-of-words (BOW) representations of 300,000 news articles from The New York Times. The vocabulary size of the text collection is 102,660. Due to the BOW encoding, \textit{NYTimes} has a very large number of dimensions and is highly sparse. To make processing feasible, we reduced the number of dimensions to 100 using terminal embeddings.

\subsubsection*{Preprocessing \& Experimental Setup}
To understand how denoising effects the quality of the outputted coresets, we applied Principal Component Analysis (PCA) on \textit{Caltech}, \textit{Census}, \textit{Covertype}, and \textit{NYTimes} by using the $k$ singular vectors corresponding to the largest singular values. 
We did not perform any preprocessing on \textit{Tower} due to its low dimensionality.

We followed the same experimental procedure with respect to the choice of parameter values for the algorithms as prior works~\cite{AckermannMRSLS12, FGSSS13}. For the target coreset size $T$, we experimented with $T=mk$ for $m = \{50, 100, 200, 500\}$. On \textit{Caltech}, \textit{Census},  \textit{Covertype}, and \textit{NYTimes}, we used values $k$ in $\{10, 20, 30, 40, 50\}$, while for \textit{Tower} we used larger cluster sizes $k \in \{20, 40, 60, 80, 100\}$. On the benchmark, we used  $k \in \{10, 20, 30, 40\}$.

We implemented Sensitivity Sampling, Group Sampling, Ray Maker, and StreamKM++ in C++. The source code can be found on GitHub\footnote{\url{https://github.com/sheikhomar/eval-k-means-coresets}}. For BICO, we used the authors' reference implementation\footnote{\url{https://ls2-www.cs.tu-dortmund.de/grav/en/bico}}. The source code was compiled with gcc 9.3.0. The experiments were performed on a machine with Intel Core i9 10940X 3.3GHz 14-Core and 2x DDR4 PC3200 128GB RAM.

\begin{figure*}
  \includegraphics[width=.65\linewidth]{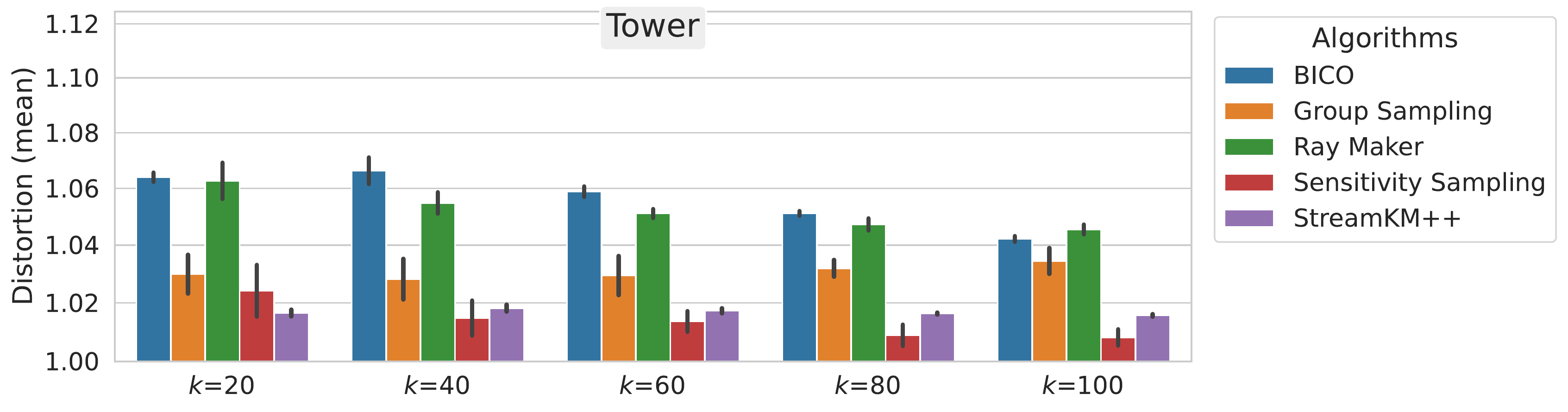}
  \newline \newline
  \subfloat{
    \includegraphics[width=.48\textwidth]{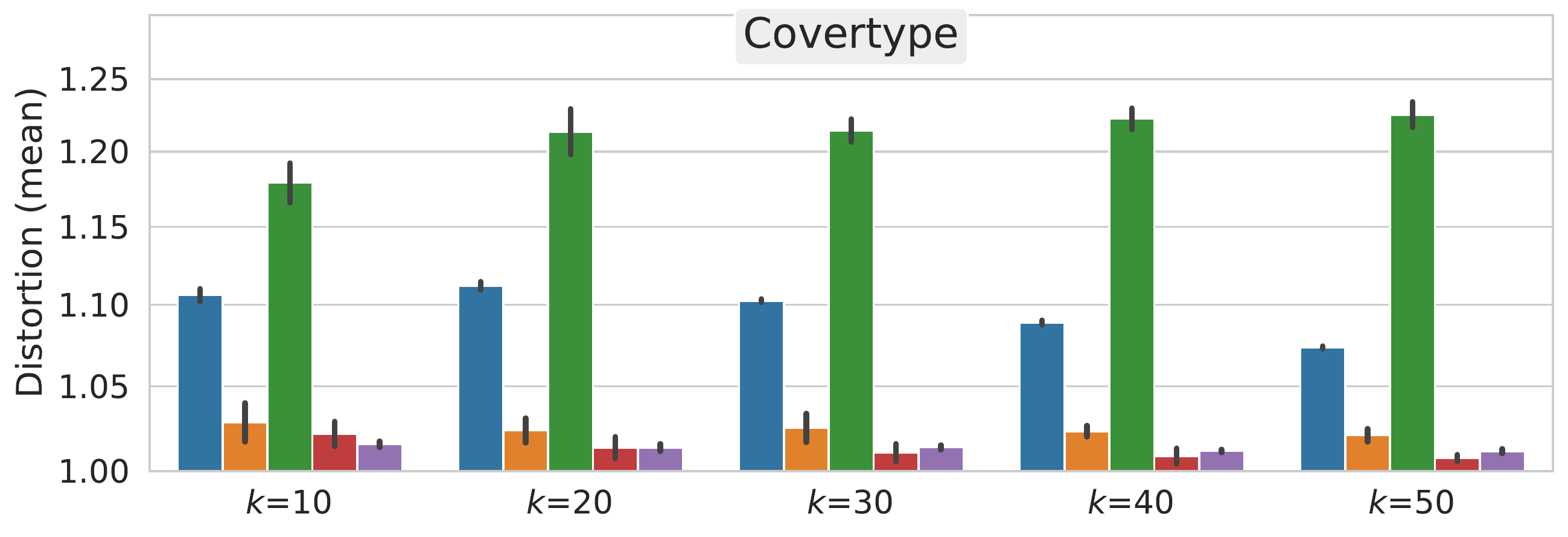}
  }
  \subfloat{
    \includegraphics[width=.48\linewidth]{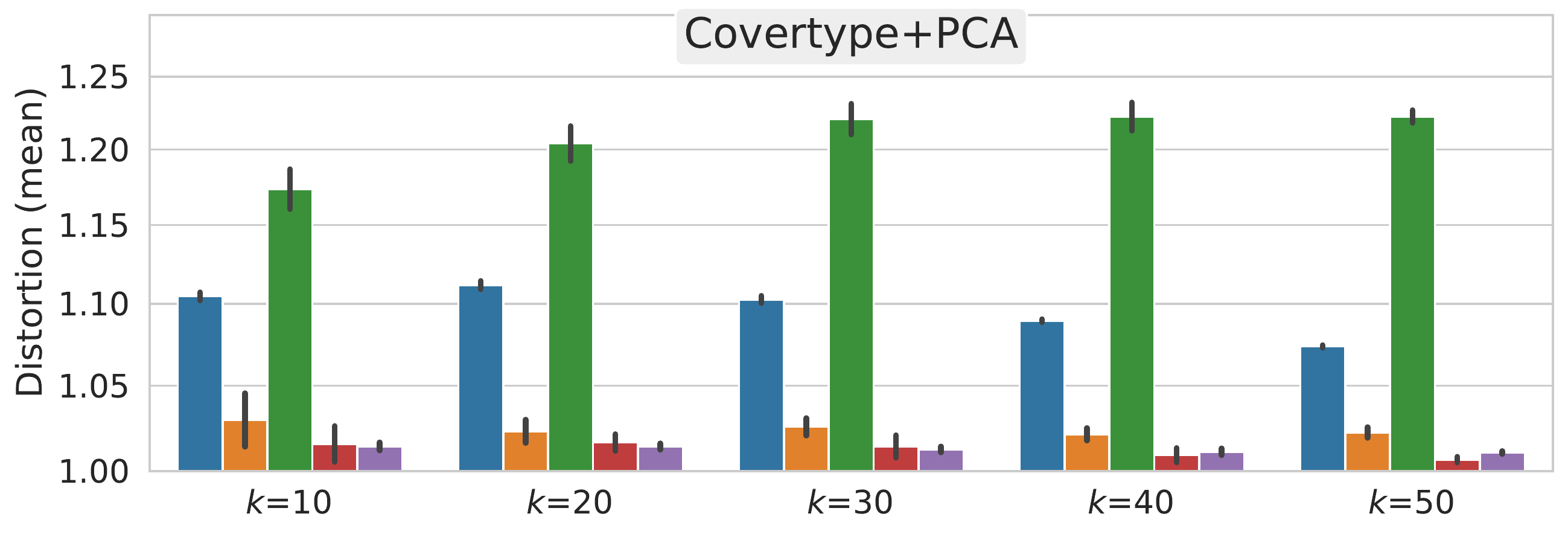}
  }
  \newline \newline
  \subfloat{
    \includegraphics[width=.48\textwidth]{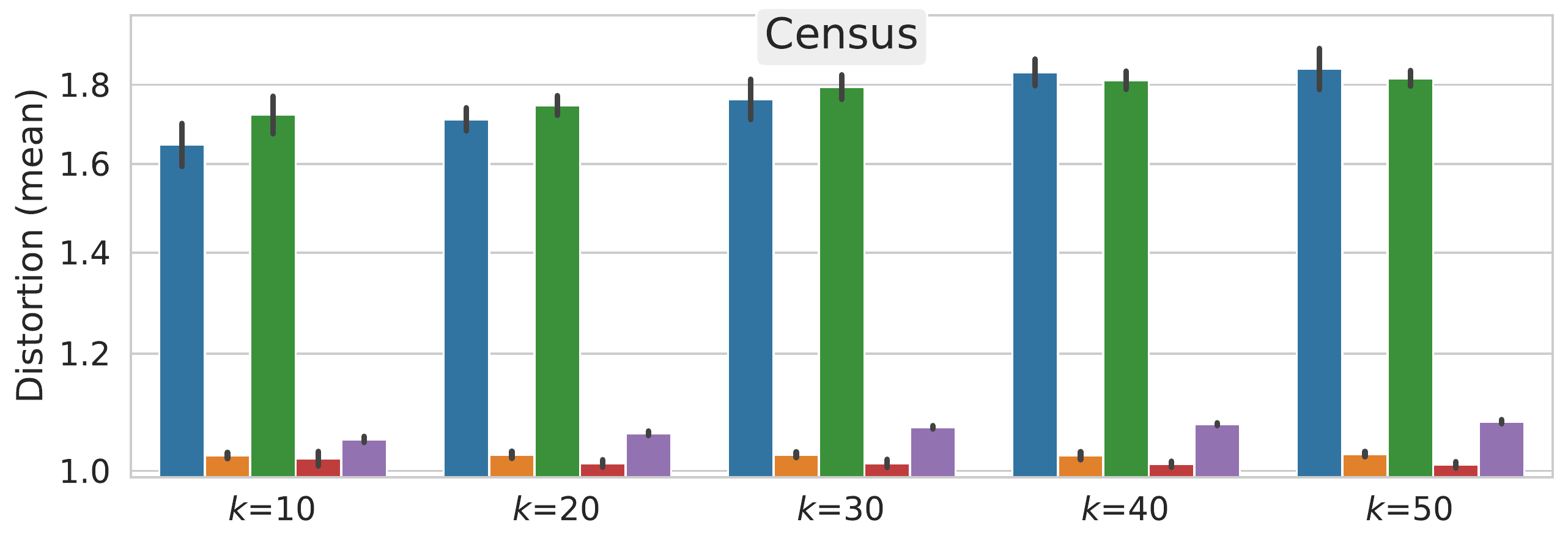}
  }
  \subfloat{
    \includegraphics[width=.48\linewidth]{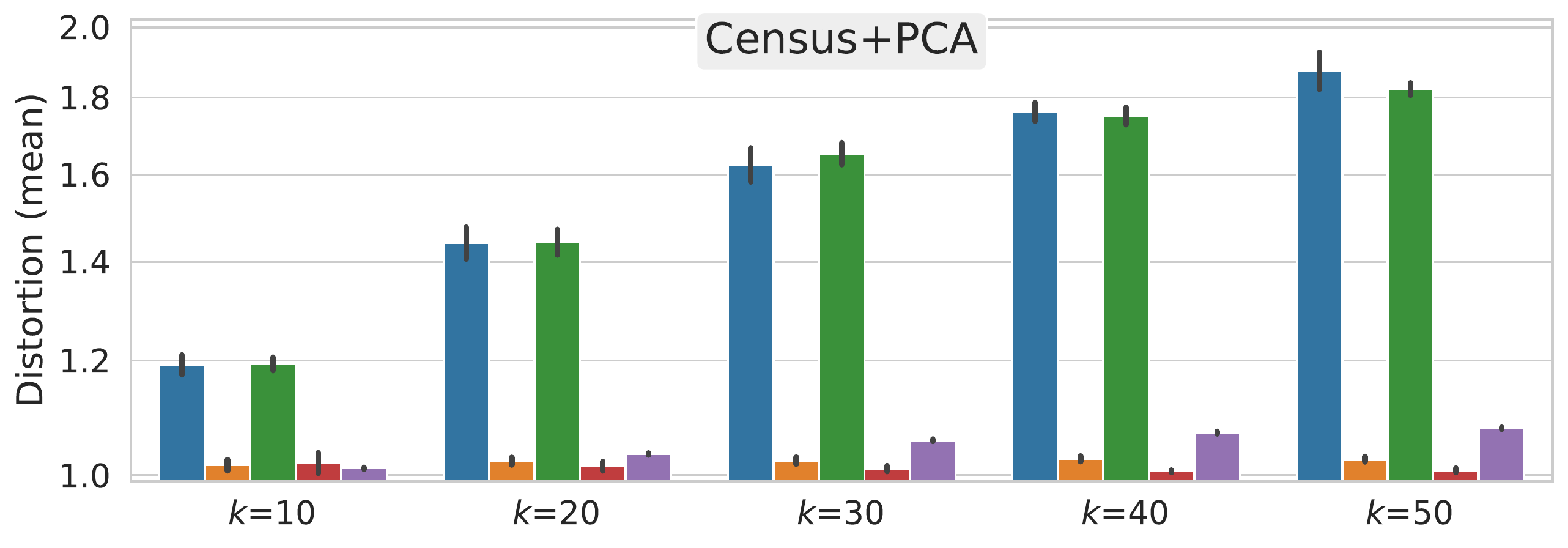}
  }
  \newline \newline
  \subfloat{
    \includegraphics[width=.48\textwidth]{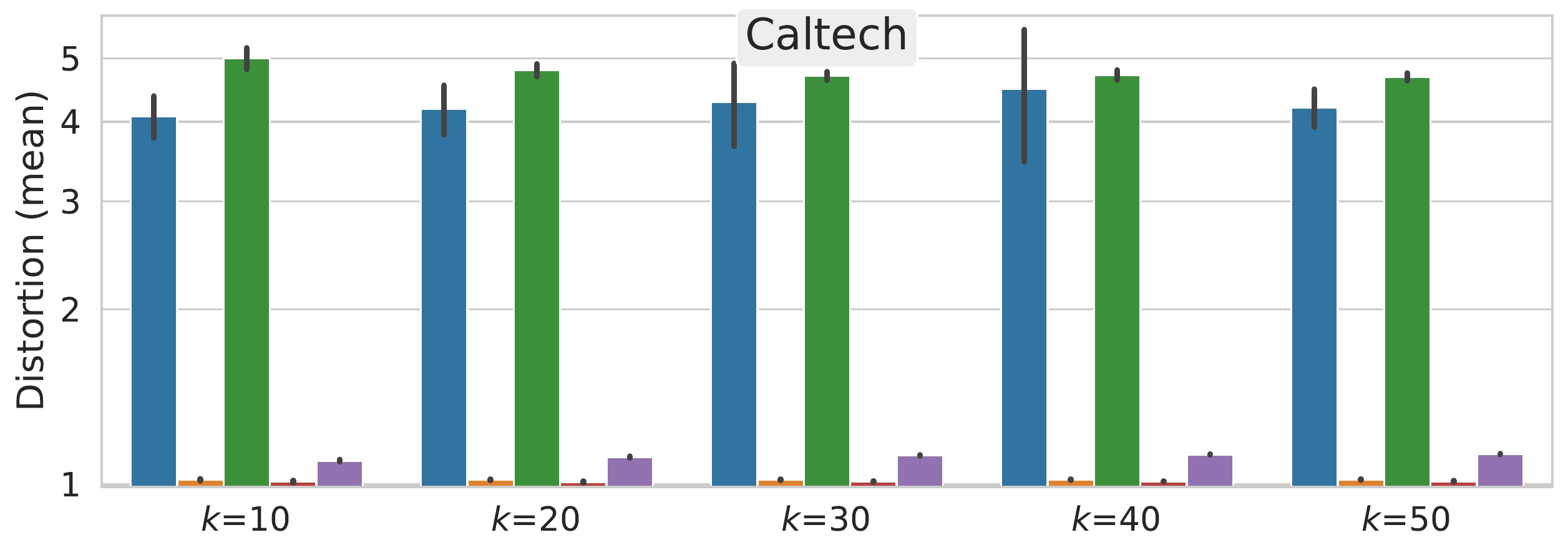}
  }
  \subfloat{
    \includegraphics[width=.48\linewidth]{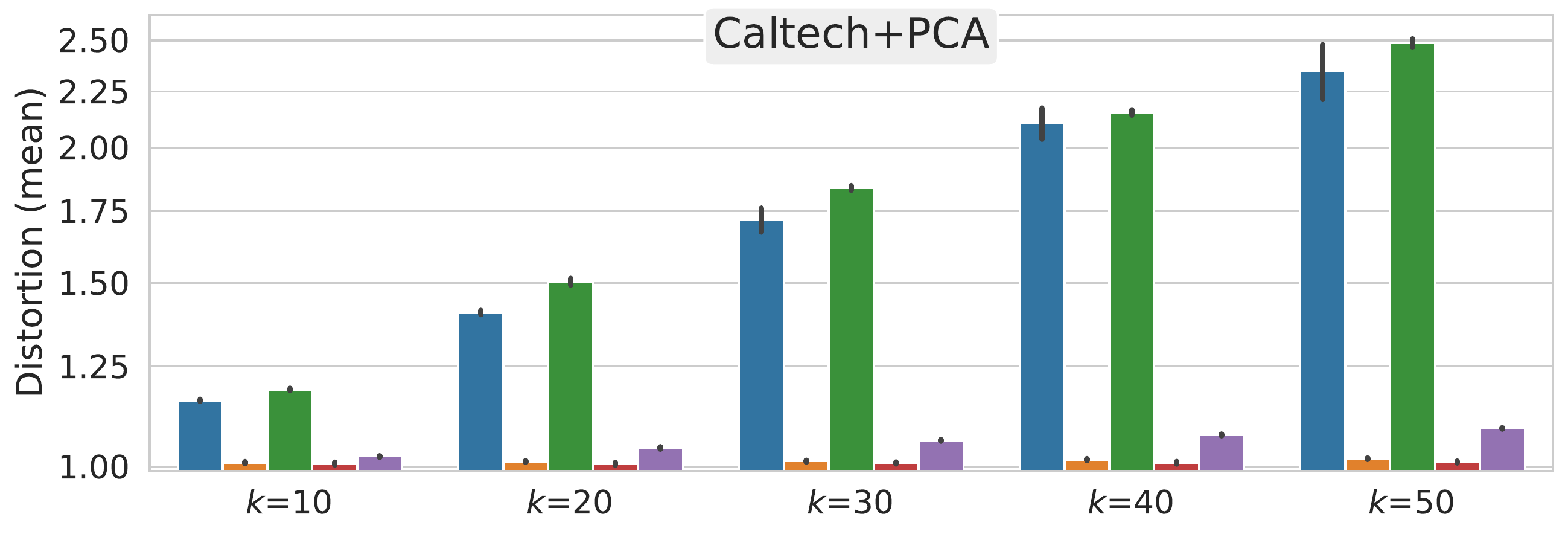}
  }
  \newline \newline
  \subfloat{
    \includegraphics[width=.48\textwidth]{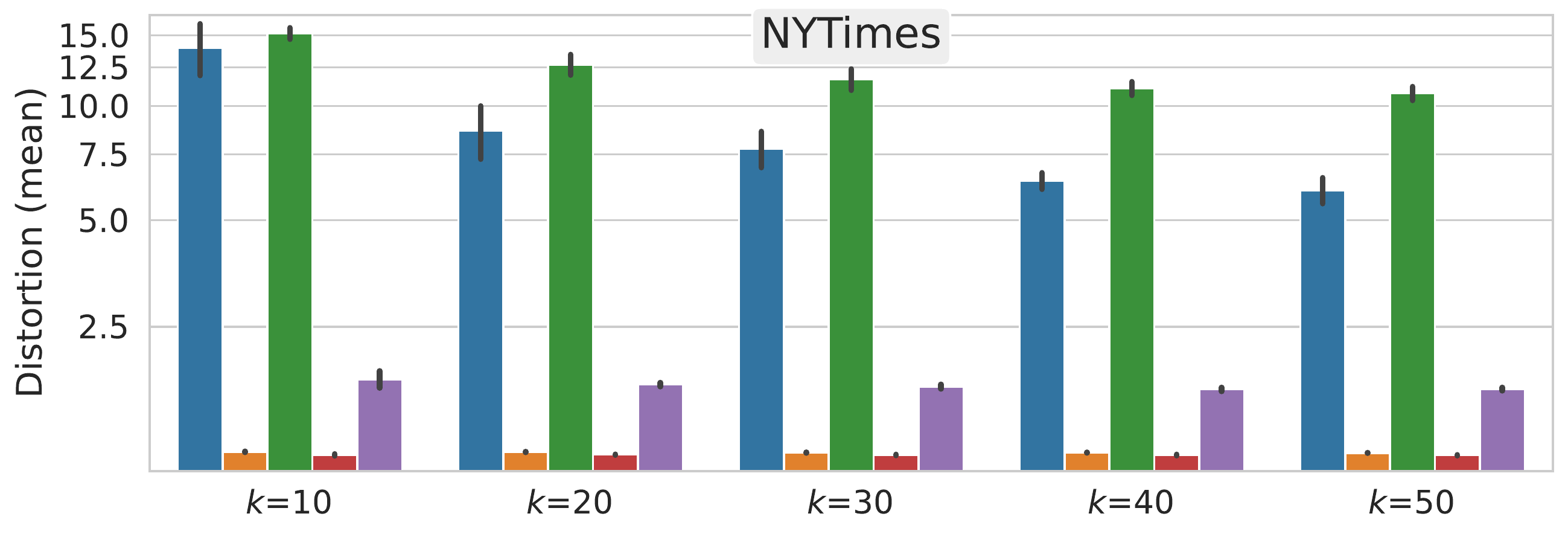}
  }
  \subfloat{
    \includegraphics[width=.48\linewidth]{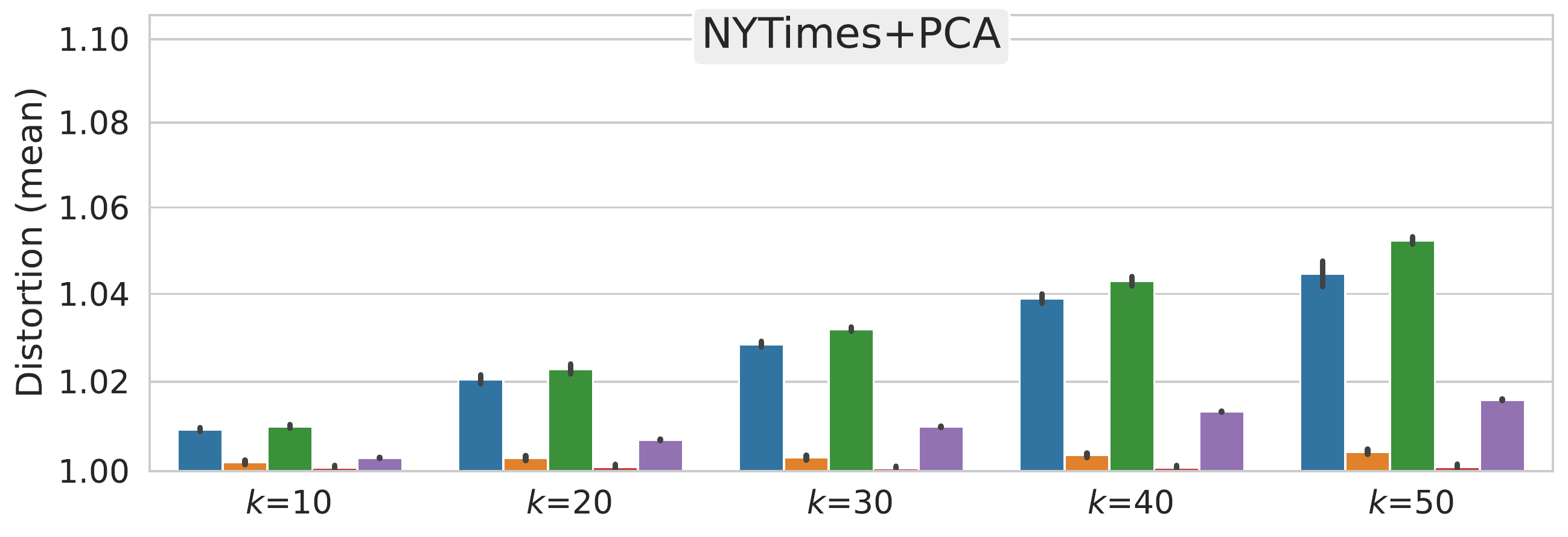}
  }
  \newline \newline
  \subfloat{
    \includegraphics[width=0.165\textwidth]{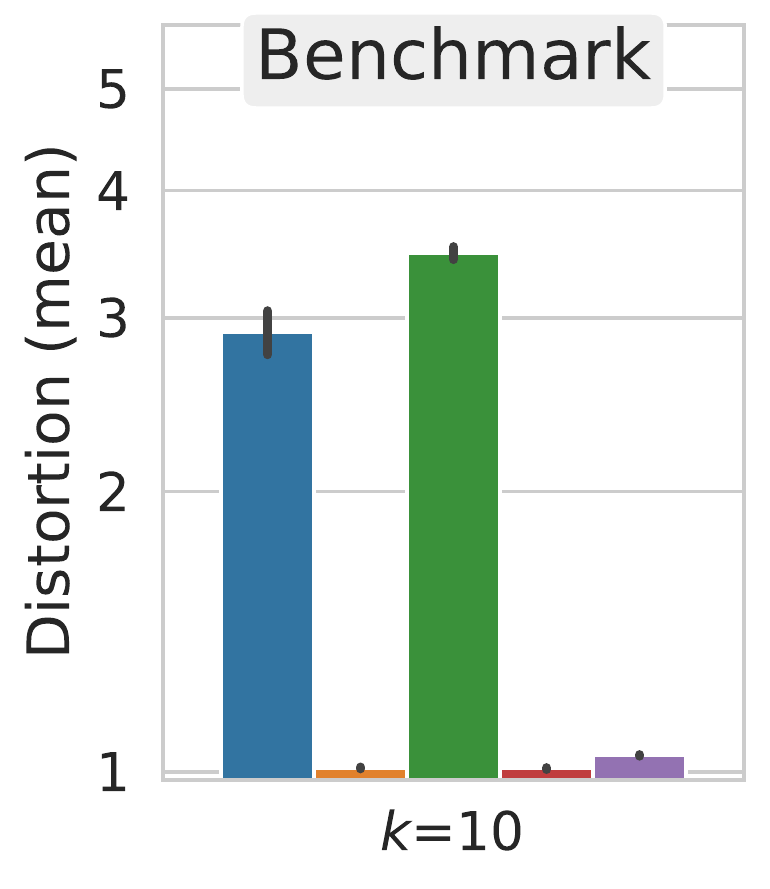}
    \includegraphics[width=0.165\textwidth]{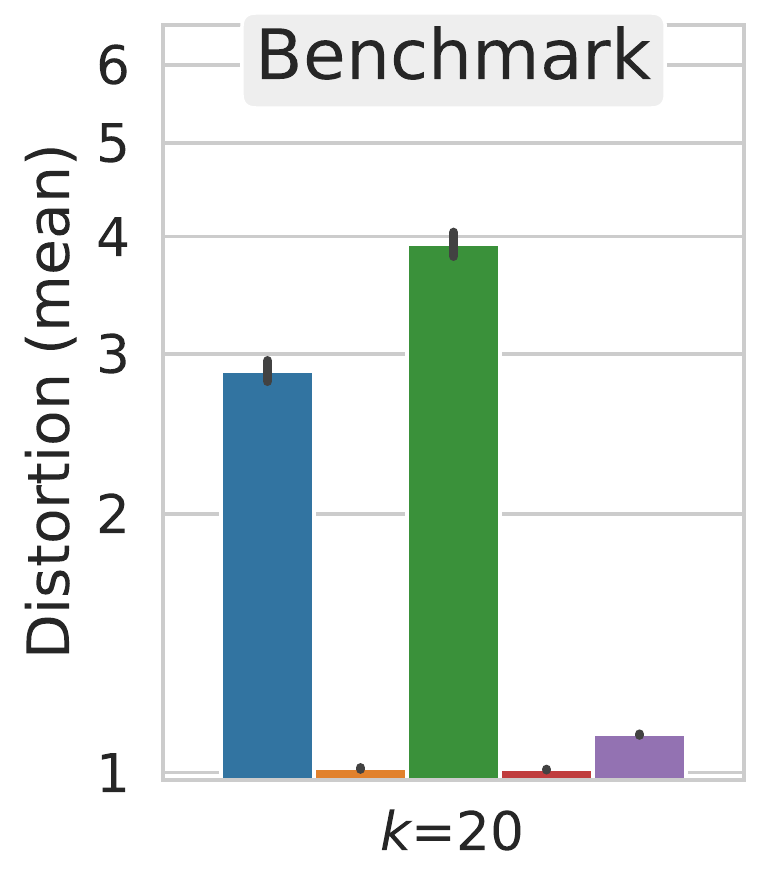}
    \includegraphics[width=0.165\textwidth]{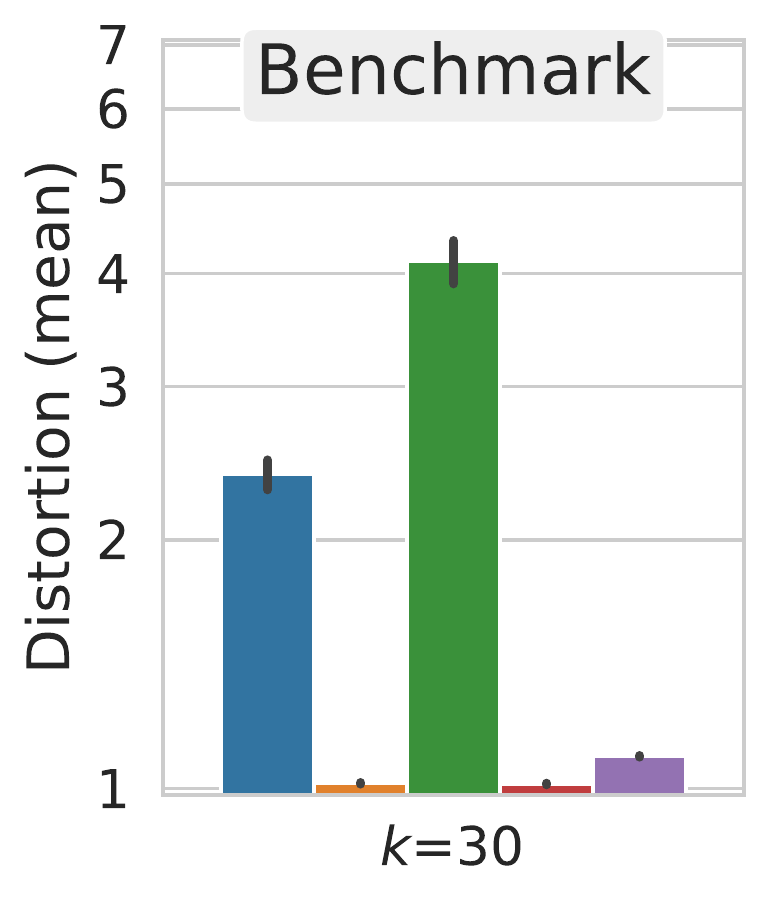}
    \includegraphics[width=0.331\textwidth]{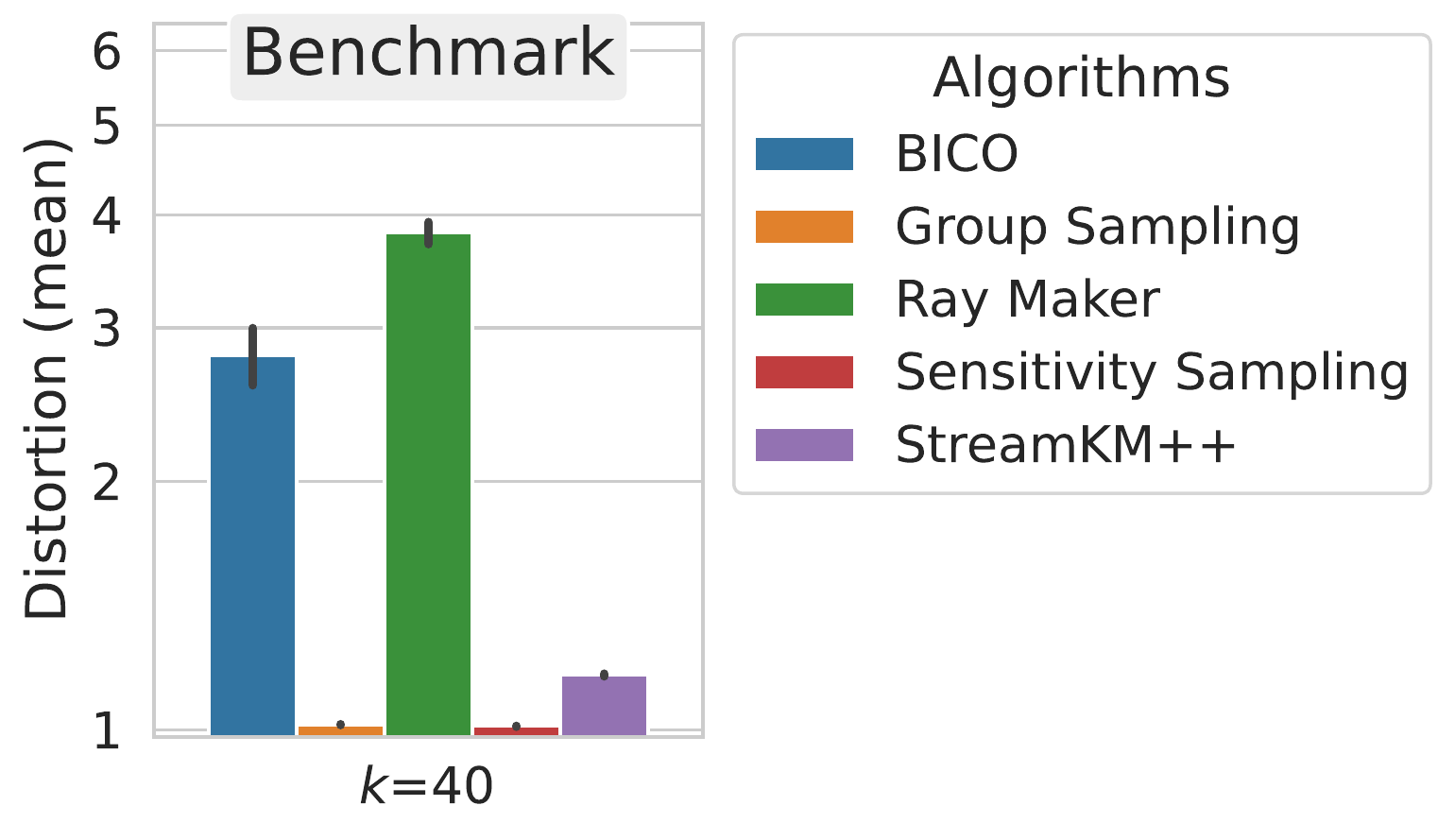}
  }
  \caption{The average distortions of the evaluated coreset algorithms with coreset size $T=200k$ on five real-world data sets and on four benchmark instances. Black bars indicate standard deviations. Notice that the axis is non-linear as otherwise the bars for Sensitivity Sampling and Group Sampling would disappear on the plots as their distortions are close to 1.}
  \label{fig:distortions}
\end{figure*}

\subsubsection*{Outcome of Experiments}
We observed that in the majority of our experiments, varying the coreset sizes does not significantly change the performance profiles of individual algorithms when comparing them against each other. Therefore, in the following sections, we focus on a cross-section of the experiments where $m=200$ i.e., coreset sizes $T=200k$.
For numerical results including variances of all the experiments and tables containing distortions, costs and running times, we refer to \cref{sec:distortions-tables}, \cref{sec:costs-tables} and \cref{sec:running-times-tables} in the appendix.

In~\cref{fig:distortions}, we summarized the distortions of the experiments with coreset sizes $T=200k$. All five algorithms are matched on the \textit{Tower} dataset. The worst distortions across the algorithms are close to 1, and performance between the algorithms is negligible. The performance difference between sampling-based and movement-based methods become more pronounced as the number of dimensions increase. On \textit{Covertype} with its 54 features, Ray Maker performs the worst followed by BICO and Group Sampling while Sensitivity Sampling and StreamKM++ perform the best. Differences in performance are more noticeable on \textit{Census}, \textit{Caltech}, and \textit{NYTimes}  where methods based on importance sampling perform much better. Sensitivity Sampling and Group Sampling perform the best, StreamKM++ come in second while BICO and Ray Maker perform the worst across these data sets.
On the \textit{Benchmark}, Ray Maker is the worst while Sensitivity Sampling and Group Sampling are the best. StreamKM++ performs also very well compared to BICO.

\subsubsection*{Interpretation of Experimental Results}

{\bf Optimization versus Compression:}
While all five algorithms are equally matched when optimizing on the candidate coresets, coreset quality performance differ significantly (see~\cref{fig:distortions}). 
For all data sets, the obtained costs differed insignificantly for all values of $k$, (see~\cref{sec:costs-tables} and \cref{fig:real-costs} in the appendix), irrespective of the coreset algorithm used, while distortions varied strongly, depending on the coreset algorithm.

Nevertheless, the cost drop with increasing values of $k$ is a predictor for the quality of certain coresets. It is not uncommon for the $k$-means cost of real-world data sets to drop significantly for larger values of $k$.
~\cref{fig:cost-curves-real-world-datasets} illustrates this behavior for several real-world data sets. The more the curve bends, the less of a difference there is between computing a coreset and a clustering with low cost. For data sets with an L-shaped cost curve, a coreset algorithm adding more centers to the coreset will seem to be performing well when evaluating it based on the outcome of the optimization.
\textit{Tower} is a good example of a data set where optimization is very close to compression. Its cost curve bends the most which indicates that adding more centers help reduce the cost. One of the strengths of the benchmark is that there is no way of reducing the cost without capturing the right subclusters within a benchmark instance. This means that the cost does not decrease markedly beyond a certain value of $k$ even if more centers are added, see~\cref{fig:cost-curves-benchmark} in the appendix. 

For BICO, Ray Maker, and StreamKM++, there is a correlation between the steepness of the cost curve for a data set and the distortion of the generated coreset. 
On data sets where the curve is less steep, we observed higher distortions. The effect is more pronounced for BICO and Ray Maker than for StreamKM++. Importance sampling approaches (Group Sampling and Sensitivity Sampling) seem to be free from this behavior as they consistently generate high quality coresets irrespective of the shape of cost curve.

\begin{figure}
  \centering
  \includegraphics[width=1\linewidth]{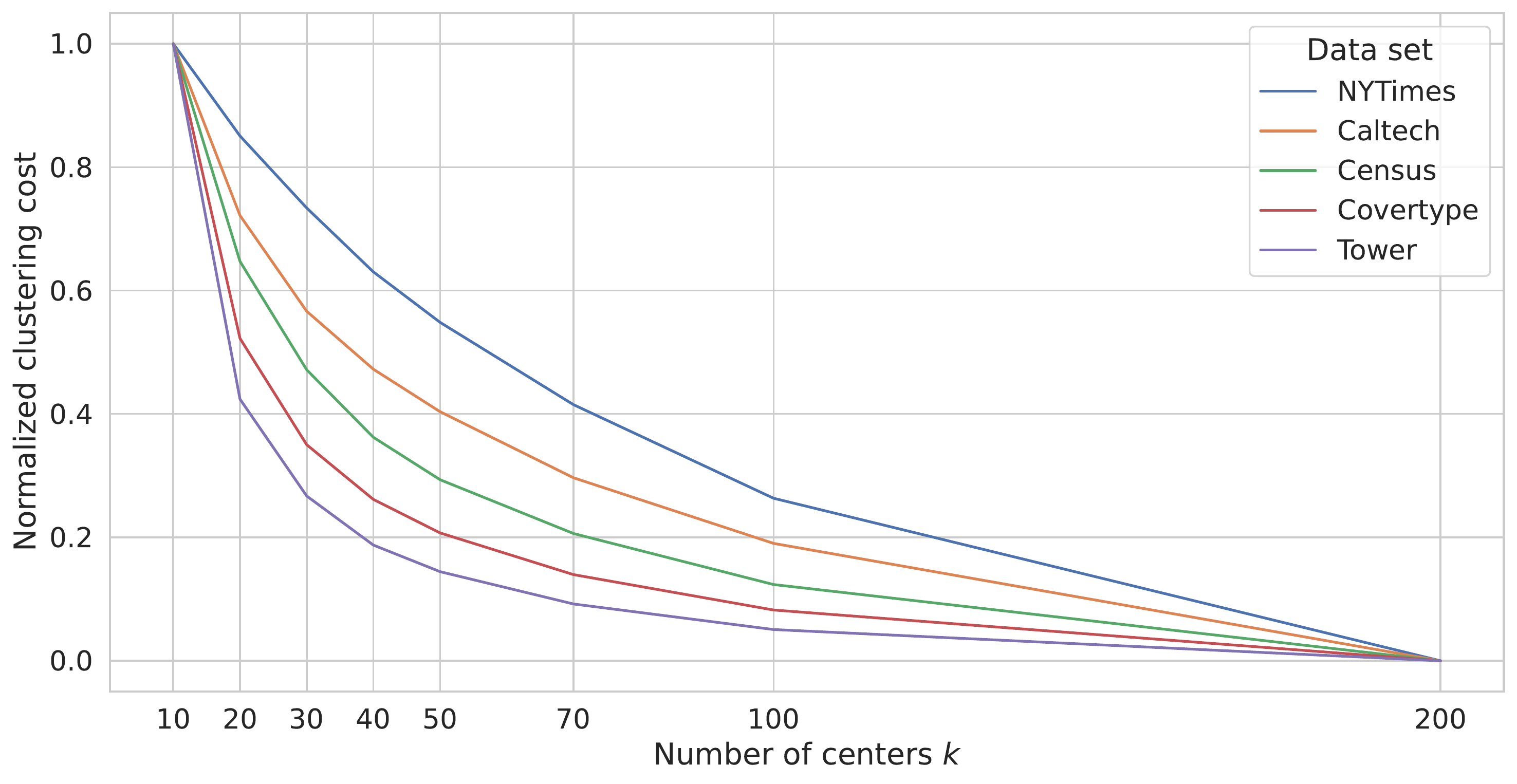}
  \caption{Depicts how clustering costs of five real-world data sets decrease as the number of centers increase. 
  Plotting the cost curve allows us to study whether we can observe a difference between coreset construction and optimization in a data set when evaluating a coreset based on cost.
  }
  \label{fig:cost-curves-real-world-datasets}
\end{figure}

{\bf Movement-based versus Sampling-based Approaches:}
In general, movement-based constructions perform the worst in terms of coreset quality. 
We observed that BICO and Ray Maker have the highest distortions across all data sets including on the benchmark instances. Among the sampling-based algorithms, Sensitive Sampling performs well with Group Sampling generally being competitive. This runs contrary to theory where Group Sampling has the better (currently known) theoretical bounds. StreamKM++ is an interesting case. Like the movement-based methods, its distortion increases with the dimension. Nevertheless, it generally performs significantly better than BICO and Ray Maker. This can be attributed to the fact that the coreset produced by StreamKM++ consists entirely of $k$-means++ centers weighted by the number of points of a minimal cost assignment. This is similar to movement-based algorithms such as BICO. Nevertheless, it also retains some of the performance from pure importance schemes.

In practice as well as in theory, the distortion of movement-based algorithms are affected by the dimension. By comparison, sampling-based algorithms are affected very little. Theoretically, there should not exist a difference, as the sampling bounds are independent of the dimension. What little effect can be observed is likely due to PCA making it easier to find low cost solutions that form the backbone of all coreset constructions. StreamKM++ is an interesting case, as it is still affected by the dimension, though less than the other movement based methods.

A notable exception is the benchmark. Here, sensitivity sampling generally found the lowest cost clustering, with BICO finding the second lowest cost clustering (see \cref{tab:real-cost-mean-std-benchmark} in the appendix). 
This happens \emph{despite} BICO generally having a worse distortion than for example Group Sampling or StreamKM++ (see \cref{tab:distortions-mean-std-benchmark} in the appendix).

{\bf Impact of PCA:}
On almost all our data sets, the performance improves when input data is preprocessed with PCA, especially for the movement-based algorithms. Empirically, the more noise is removed (i.e., small $k$ value), the lower the distortion. Notice that $k$ is the number of principal components that the input data is projected on to. The rest of the low variance components are treated as noise and removed. Method utilizing sampling (Group Sampling, Sensitivity Sampling and StreamKM++) are less effected by the preprocessing step. On \textit{Covertype}, PCA does not change the distortions by much because almost all the variance in the data is explained by the first five principal components (see~\cref{fig:explained-variance-pca}). 
On \textit{Caltech} and \textit{NYTimes}, the quality of the coresets by BICO and Ray Maker improves greatly because the noise removal is more aggressive. Even if the quality is much better for movement-based coreset constructions due to PCA, importance sampling methods are still superior when it comes to the quality of the compression. Summarizing, all methods benefit from PCA, and in case of movement-based constructions, we consider PCA a necessary preprocessing step. For the sampling-based methods, the computational expense of using PCA in preprocessing does not seem justify the comparatively meager gains in coreset distortion.

\section{Conclusion} \label{sec:conclusion}
In this work, we studied how to assess the quality of $k$-means coresets computed by state-of-the-art algorithms. 
Previous work generally measured the quality of optimization algorithms run on the coreset, which we empirically observed to be a poor indicator of coreset quality.
For real-world data sets, we sampled candidate clusterings and evaluated the worst case distortion on them. Complementing this, we also proposed a benchmark framework which generates hard instances for known $k$-means coreset algorithms. Our experiments indicate a general advantage for algorithms based on importance sampling over movement-based methods. Despite movement-based methods running on very efficient code, it is necessary to complement them with rather expensive dimension reduction methods, rendering what efficiency they might have over importance sampling somewhat moot.

Two results bear further investigation. First, the currently known provable coreset sizes for Sensitivity Sampling are worse than those provable via Group Sampling. Empirically, we observed the opposite: While Group Sampling is competitive, Sensitivity Sampling always outperforms it. Since Group Sampling requires somewhat cumbersome computational overhead, practical applications should prefer Sensitivity Sampling. In light of these results, a theoretical analysis for Sensitivity Sampling matching the performance of Group Sampling would be welcome.

The second point of interest focuses on the performance of StreamKM++. The distortion of this algorithm is significantly better than what one would expect from its theoretical analysis.
Empirically, StreamKM++ is notably better than the other movement-based constructions across all data sets, and especially on high dimensional data.
While it is not competitive to the pure importance sampling algorithms, there are several reasons for investigating it further. It essentially only requires running $k$-means++ for additional iterations, which is already a nearly ubiquitous algorithm for the $k$-means problem. Although the other sampling-based coreset algorithms can also be readily implemented, doing so might be cumbersome. In particular, the theoretically (but not empirically) best algorithm Group Sampling requires extensive preprocessing steps.
This begs the question whether there exist a better theoretical analysis for StreamKM++.

In addition, StreamKM++ currently weighs each point by the number of points assigned to it. It may also be possible to improve the performance of the algorithm in both theory and practice by using a different weighting scheme. 
We leave this as an open problem for future research.

\bibliography{main}
\newpage
\appendix
\section{Properties of the Benchmark}

\begin{fact}
For $a\neq a'$, we have $d(\mathcal{C}^{a},\mathcal{C}^{a'}) = 1-1/k$.
\end{fact}
\begin{proof}
Consider an arbitrary vector $v_i^{\ell}$. By construction, the positive entries of $v_i^{\ell}$ range from $k^{\ell-1}\cdot i+1$ to $k^{\ell-1}\cdot (i+1)$. Similarly, the positive entries for the vector $v_j^{\ell-1}$ range from range from $k^{\ell-2}\cdot j+1$ to $k^{\ell-2}\cdot (j+1)$. Therefore, concatenating $v_j^{\ell-1}$ $k$ times into a vector $v'$, $v'$ and $v_i^{\ell}$ can share at most one positive coordinate. Inductively, the same holds true for any concatenation of vectors $v_j^{\ell-h}$.
Thus, the two clusters induced by the columns formed by concatenating the vectors $v$ can share only a $1/k$ fraction of the points. Since each cluster consists of exactly $k^{\alpha}/k$ = $k^{\alpha-1}$ points, the confusion matrix $M$ only has entries $\frac{n}{k^2}$ and for any permutation $\pi$, we have $d(\mathcal{C}^{a},\mathcal{C}^{a'}) = 1-1/k$.
\end{proof}

\begin{fact}
\label{fact:cost}
For any $C_j^a$, we have $\cost_{C_j^a}(\{\mu(C_j^a)\}) = (\alpha-1)\cdot k^{\alpha-2}\cdot (k-1)$.
\end{fact}
\begin{proof}
Without loss of generality, we consider $C_1^0$; the proof is analogous for the other choices of $j$ and $a$. We first note that for any point $A_i \in C_1^0$, the coordinates $A_{i,\ell}$ are identical for $\ell <k$. Furthermore for the column $\ell\geq k$, we have by construction $\sum_{A_i\in C_j} A_{i,\ell} = k^{\alpha-1}\cdot \frac{k-1}{k} + (k^{\alpha}-k^{\alpha-1})\frac{1}{k}=k^{\alpha-1}\cdot (\frac{k-1}{k} - (k-1)\frac{1}{k}) = 0.$ Therefore, the mean of $C_1^0$ satisfies $\mu(C_1^0)_{\ell} = \begin{cases}A_{i,\ell} &\text{if }\ell<k \\
0 &\text{else.}\end{cases}$. 
Thus, the cost is precisely $(\alpha-1)\cdot k^{\alpha-1}\cdot \left(\left(\frac{k-1}{k}\right)^2 + \left(\frac{1}{k}\right)^2 \right)=(\alpha-1)\cdot k^{\alpha-2}\cdot (k-1)$.
\end{proof}

Finally, we show that the means for the clustering $\mathcal{C}^{a}$ also induce $\mathcal{C}^{a}$.

\begin{fact}
\label{fact:opt}
For a clustering $\mathcal{C}^{a}$, let $\mu(C_j^{a})$ denote the mean of cluster $C_j^a$. Then every point  is assigned to its closest center. Moreover, every point $A_i$ of $C_j^a$ has equal distance to any center $\mu(C_h^{a})$ with $h\neq j$.
\end{fact}
\begin{proof}
Again, we assume without loss of generality $a=0$.
Let $A_i$ be an arbitrary point of cluster $C_{h}^{a}$ and consider the mean $\mu(C_j^a)_{\ell} = \begin{cases}A_{i,\ell} &\text{if }\ell<k \\
0 &\text{else.}\end{cases}$ of cluster $C_j^a$. By definition, the positive coordinates of $A_i$ are not equal to the positive coordinates of $\mu(C_j^a)$. The only difference in coordinates between the means of $\mu(C_j^a)$ and $\mu(C_h^a)$ are the first $k$ coordinates, as the rest are $0$.
But here the coordinates of $\mu(C_h^a)$ and $A_i$ are identical, hence $\mu(C_j^a)$ cannot be closer to $A_i$.

To prove that the distances between $A_i$ and any $\mu(C_h^{a})$ with $h\neq j$ are equal, again consider that any difference can only exist among the first $k$ coordinates. Here, we have $\mu(C_h^{a})_h = \frac{k-1}{k}$, and the remaining columns are $-\frac{1}{k}$. Since $A_{i,h} = -\frac{1}{k}$ for any $h\neq j$, the claim follows.
\end{proof}

\section{Hardness of Weak Coreset Evaluation}

Here, we show that it is in general co-NP hard to evaluate whether two point sets $A$ and $B$ are weak coresets of each other. A weak coreset only requires that a $(1+\varepsilon)$ approximation for one point set is a $(1+O(\varepsilon))$ for the other.

\begin{proposition}
\label{prop:hardness}
Given two point sets $A$ and $B$ in $\mathbb{R}^d$ and a sufficiently small (constant) $\varepsilon>0$, it is co-NP hard to decide whether $A$ is a weak coreset of $B$.
\end{proposition}
\begin{proof}
First, we recall that for some $\varepsilon_0$ and candidate clustering cost $V$, it is NP-hard to decide whether there exists a clustering $C$ with cost in $\cost_A(C)\leq V$ and $\cost_B(C)\geq (1+\varepsilon_0)\cdot V$.
Conversely, it is co-NP-hard to decide whether there exists no set of centers $C$ such that $\cost_A(C)\leq V$ and $\cost_B(C) \geq (1+\varepsilon_0)\cdot V$.
\end{proof}

We remark that the possible values for $\varepsilon_0$ are determined by the current APX-hardness results. Assuming NP$\neq$P, $\varepsilon_0\approx 1.07$ and assuming UCG, $\varepsilon_0 \approx 1.17$~\cite{Cohen-AddadSL21,Cohen-AddadS19} for $k$-means in Euclidean spaces.

\paragraph*{Further Extensions}

On the benchmark we considered here, both Sensitivity Sampling, as well as Group Sampling are similar to uniform sampling and, indeed, uniform sampling could be used to construct a good coreset. We can eliminate uniform sampling as a viable algorithm for this instance by combining multiple benchmarks $B_1,\ldots B_t$ with $\sum_{i=1}^t  k_i =k$. Each benchmark then has size $\sum_{i=1}^t  k_i^{\alpha}$. We then add an additive offset to the coordinates of each benchmark such that they do not interfere. In this case, uniform sampling does not work if the values of the $k_i$ are different enough. Since it is well known that uniform sampling is not a viable coreset algorithm in both theory and practice, we only used the basic benchmark for our evaluations.

\section{Distortions}
\label{sec:distortions-tables}

The tables 
(\cref{tab:distortions-mean-std-benchmark},
 \cref{tab:distortions-mean-std-caltech},
 \cref{tab:distortions-mean-std-caltech-pca},
 \cref{tab:distortions-mean-std-census},
 \cref{tab:distortions-mean-std-census-pca},
 \cref{tab:distortions-mean-std-covertype},
 \cref{tab:distortions-mean-std-covertype-pca},
 \cref{tab:distortions-mean-std-nytimes},
 \cref{tab:distortions-mean-std-nytimes-pca},
 \cref{tab:distortions-mean-std-tower})
show the distortions of the 5 evaluated algorithms on the different data sets.
We vary the coreset size $T$ for different $k$ values using the formula: $T=mk$ where $m = \{50, 100, 200, 500\}$.
The running time for StreamKM++ with coreset size $T=500k$ exceeds the allocated time budget of 32 hours on almost all data sets.
For this reason, the distortions for StreamKM++ with $m=500$ are not present in the tables.

\begin{longtable}{lllllll}
\multicolumn{7}{c}{\textbf{Distortions on the \textit{Caltech} data set}} \\
\toprule
\parbox[t]{5mm}{\ \\$k$} 
& \parbox[t]{5mm}{\ \\$m$} 
& BICO 
& \parbox[t]{1.7cm}{Group\\Sampling} 
& \parbox[t]{1.7cm}{Ray\\Maker}
& \parbox[t]{1.7cm}{Sensitivity\\Sampling}
&    StreamKM++ \\
\midrule
10 & 50  &  3.40 (0.440) &   1.02 (0.010) &  5.05 (0.157) &         1.02 (0.005) &  1.07 (0.005) \\
   & 100 &  3.24 (0.729) &   1.01 (0.004) &  3.84 (0.081) &         1.01 (0.003) &  1.05 (0.004) \\
   & 200 &  2.90 (0.153) &   1.01 (0.002) &  3.48 (0.052) &         1.01 (0.002) &  1.04 (0.002) \\
   & 500 &  2.62 (0.095) &   1.01 (0.001) &  3.40 (0.058) &         1.00 (0.001) &  \\
 \midrule
20 & 50  &  3.22 (0.160) &   1.04 (0.004) &  5.52 (0.266) &         1.02 (0.003) &  1.08 (0.006) \\
   & 100 &  3.09 (0.122) &   1.02 (0.004) &  4.31 (0.130) &         1.01 (0.002) &  1.08 (0.003) \\
\end{longtable}

\begin{longtable}{lllllll}
\multicolumn{7}{c}{\textbf{Distortions on the \textit{Caltech} data set}} \\
\toprule
\parbox[t]{5mm}{\ \\$k$} 
& \parbox[t]{5mm}{\ \\$m$} 
& BICO 
& \parbox[t]{1.7cm}{Group\\Sampling} 
& \parbox[t]{1.7cm}{Ray\\Maker}
& \parbox[t]{1.7cm}{Sensitivity\\Sampling}
&    StreamKM++ \\
\midrule
10 & 50  &  5.12 (0.307) &   1.07 (0.009) &  7.05 (0.250) &         1.05 (0.013) &  1.16 (0.008) \\
   & 100 &  4.48 (0.284) &   1.04 (0.006) &  5.46 (0.254) &         1.03 (0.008) &  1.13 (0.005) \\
   & 200 &  4.08 (0.319) &   1.02 (0.004) &  5.00 (0.196) &         1.01 (0.005) &  1.10 (0.005) \\
   & 500 &  3.41 (0.215) &   1.01 (0.003) &  4.75 (0.188) &         1.00 (0.003) &  \\
 \midrule
20 & 50  &  6.35 (1.173) &   1.06 (0.008) &  6.92 (0.209) &         1.05 (0.007) &  1.18 (0.008) \\
   & 100 &  4.65 (0.283) &   1.04 (0.005) &  5.29 (0.146) &         1.03 (0.005) &  1.14 (0.005) \\
\end{longtable}

\begin{longtable}{lllllll}
\multicolumn{7}{c}{\textbf{Distortions on the \textit{Caltech} data set (with PCA preprocessing)}} \\
\toprule
\parbox[t]{5mm}{\ \\$k$} & \parbox[t]{5mm}{\ \\$m$} &BICO & \parbox[t]{1.7cm}{Group\\Sampling} &\parbox[t]{1.7cm}{Ray\\Maker}&\parbox[t]{1.7cm}{Sensitivity\\Sampling}&    StreamKM++ \\
\midrule
10 & 50  &  1.21 (0.009) &   1.03 (0.004) &  1.25 (0.005) &         1.02 (0.007) &  1.04 (0.005) \\
   & 100 &  1.19 (0.006) &   1.02 (0.006) &  1.21 (0.006) &         1.01 (0.004) &  1.03 (0.004) \\
   & 200 &  1.16 (0.003) &   1.01 (0.002) &  1.19 (0.003) &         1.01 (0.003) &  1.02 (0.002) \\
   & 500 &  1.13 (0.007) &   1.00 (0.002) &  1.19 (0.006) &         1.00 (0.001) &  \\
 \midrule
20 & 50  &  1.54 (0.024) &   1.04 (0.006) &  1.63 (0.008) &         1.03 (0.006) &  1.07 (0.003) \\
   & 100 &  1.48 (0.020) &   1.02 (0.002) &  1.54 (0.010) &         1.01 (0.004) &  1.06 (0.003) \\
\end{longtable}

\begin{longtable}{lllllll}
\multicolumn{7}{c}{\textbf{Distortions on the \textit{Census} data set}} \\
\toprule
\parbox[t]{5mm}{\ \\$k$} & \parbox[t]{5mm}{\ \\$m$} &BICO & \parbox[t]{1.7cm}{Group\\Sampling} &\parbox[t]{1.7cm}{Ray\\Maker}&\parbox[t]{1.7cm}{Sensitivity\\Sampling}&    StreamKM++ \\
\midrule
10 & 50  &  2.08 (0.106) &   1.08 (0.019) &  2.30 (0.162) &         1.03 (0.031) &  1.10 (0.012) \\
   & 100 &  1.83 (0.124) &   1.04 (0.021) &  1.92 (0.094) &         1.03 (0.012) &  1.07 (0.006) \\
   & 200 &  1.65 (0.056) &   1.02 (0.005) &  1.72 (0.050) &         1.02 (0.012) &  1.05 (0.005) \\
   & 500 &  1.48 (0.053) &   1.02 (0.004) &  1.69 (0.029) &         1.02 (0.010) &  \\
 \midrule
20 & 50  &  2.18 (0.164) &   1.07 (0.022) &  2.31 (0.078) &         1.03 (0.007) &  1.10 (0.008) \\
   & 100 &  1.93 (0.080) &   1.04 (0.011) &  1.91 (0.048) &         1.02 (0.011) &  1.08 (0.005) \\
\end{longtable}

\begin{longtable}{lllllll}
\multicolumn{7}{c}{\textbf{Distortions on the \textit{Census} data set (with PCA preprocessing)}} \\
\toprule
\parbox[t]{5mm}{\ \\$k$} & \parbox[t]{5mm}{\ \\$m$} &BICO & \parbox[t]{1.7cm}{Group\\Sampling} &\parbox[t]{1.7cm}{Ray\\Maker}&\parbox[t]{1.7cm}{Sensitivity\\Sampling}&    StreamKM++ \\
\midrule
10 & 50  &  1.34 (0.048) &   1.06 (0.018) &  1.35 (0.030) &         1.04 (0.021) &  1.03 (0.007) \\
   & 100 &  1.24 (0.015) &   1.03 (0.007) &  1.25 (0.017) &         1.03 (0.016) &  1.02 (0.005) \\
   & 200 &  1.19 (0.020) &   1.02 (0.009) &  1.19 (0.013) &         1.02 (0.018) &  1.01 (0.002) \\
   & 500 &  1.11 (0.003) &   1.01 (0.005) &  1.18 (0.017) &         1.01 (0.003) &  \\
 \midrule
20 & 50  &  1.84 (0.118) &   1.07 (0.018) &  1.78 (0.052) &         1.02 (0.012) &  1.07 (0.005) \\
   & 100 &  1.57 (0.073) &   1.03 (0.010) &  1.56 (0.020) &         1.02 (0.008) &  1.05 (0.005) \\
\end{longtable}

\begin{longtable}{lllllll}
\multicolumn{7}{c}{\textbf{Distortions on the \textit{Covertype} data set}} \\
\toprule
\parbox[t]{5mm}{\ \\$k$} & \parbox[t]{5mm}{\ \\$m$} &BICO & \parbox[t]{1.7cm}{Group\\Sampling} &\parbox[t]{1.7cm}{Ray\\Maker}&\parbox[t]{1.7cm}{Sensitivity\\Sampling}&    StreamKM++ \\
\midrule
10 & 50  &  1.25 (0.018) &   1.08 (0.022) &  1.31 (0.027) &         1.04 (0.025) &  1.05 (0.008) \\
   & 100 &  1.17 (0.015) &   1.04 (0.016) &  1.22 (0.017) &         1.03 (0.014) &  1.02 (0.003) \\
   & 200 &  1.11 (0.004) &   1.03 (0.012) &  1.18 (0.014) &         1.02 (0.008) &  1.02 (0.002) \\
   & 500 &  1.05 (0.001) &   1.01 (0.004) &  1.16 (0.009) &         1.01 (0.006) &  \\
 \midrule
20 & 50  &  1.27 (0.025) &   1.09 (0.022) &  1.35 (0.031) &         1.04 (0.008) &  1.05 (0.006) \\
   & 100 &  1.18 (0.004) &   1.05 (0.016) &  1.24 (0.010) &         1.02 (0.006) &  1.03 (0.003) \\
\end{longtable}

\begin{longtable}{lllllll}
\multicolumn{7}{c}{\textbf{Distortions on the \textit{Covertype} data set (with PCA preprocessing)}} \\
\toprule
\parbox[t]{5mm}{\ \\$k$} & \parbox[t]{5mm}{\ \\$m$} &BICO & \parbox[t]{1.7cm}{Group\\Sampling} &\parbox[t]{1.7cm}{Ray\\Maker}&\parbox[t]{1.7cm}{Sensitivity\\Sampling}&    StreamKM++ \\
\midrule
10 & 50  &  1.26 (0.024) &   1.09 (0.025) &  1.30 (0.019) &         1.04 (0.017) &  1.04 (0.008) \\
   & 100 &  1.16 (0.010) &   1.05 (0.017) &  1.22 (0.015) &         1.03 (0.017) &  1.02 (0.003) \\
   & 200 &  1.10 (0.003) &   1.03 (0.017) &  1.17 (0.014) &         1.02 (0.011) &  1.01 (0.002) \\
   & 500 &  1.06 (0.001) &   1.01 (0.006) &  1.16 (0.008) &         1.01 (0.007) &  \\
 \midrule
20 & 50  &  1.27 (0.018) &   1.10 (0.021) &  1.34 (0.021) &         1.04 (0.018) &  1.05 (0.004) \\
   & 100 &  1.17 (0.005) &   1.05 (0.011) &  1.25 (0.013) &         1.02 (0.010) &  1.03 (0.002) \\
\end{longtable}

\begin{longtable}{lllllll}
\multicolumn{7}{c}{\textbf{Distortions on the \textit{NYTimes} data set}} \\
\toprule
\parbox[t]{5mm}{\ \\$k$} & \parbox[t]{5mm}{\ \\$m$} & BICO & \parbox[t]{1.8cm}{Group\\Sampling} & \parbox[t]{1.9cm}{Ray\\Maker}&\parbox[t]{1.8cm}{Sensitivity\\Sampling}&    StreamKM++ \\
\midrule
10 & 50  &  35.32 (8.376) &   1.08 (0.009) &  28.37 (2.761) &         1.05 (0.012) &  2.04 (0.244) \\
   & 100 &  24.50 (5.355) &   1.06 (0.008) &  17.64 (0.561) &         1.03 (0.007) &  1.80 (0.069) \\
   & 200 &  13.97 (2.143) &   1.04 (0.004) &  15.17 (0.512) &         1.02 (0.008) &  1.75 (0.110) \\
   & 500 &   8.09 (1.241) &   1.03 (0.004) &  14.16 (1.031) &         1.01 (0.004) &  \\
 \midrule
20 & 50  &  21.70 (3.965) &   1.09 (0.008) &  23.94 (1.466) &         1.05 (0.010) &  1.97 (0.083) \\
   & 100 &  16.00 (3.906) &   1.05 (0.005) &  14.56 (0.846) &         1.03 (0.006) &  1.79 (0.046) \\
\end{longtable}

\begin{longtable}{lllllll}
\multicolumn{7}{c}{\textbf{Distortions on the \textit{NYTimes} data set (with PCA preprocessing)}} \\
\toprule
\parbox[t]{5mm}{\ \\$k$} & \parbox[t]{5mm}{\ \\$m$} &BICO & \parbox[t]{1.7cm}{Group\\Sampling} &\parbox[t]{1.7cm}{Ray\\Maker}&\parbox[t]{1.7cm}{Sensitivity\\Sampling}&    StreamKM++ \\
\midrule
10 & 50  &  1.02 (0.002) &   1.01 (0.002) &  1.02 (0.001) &         1.00 (0.002) &  1.01 (0.000) \\
   & 100 &  1.01 (0.001) &   1.00 (0.001) &  1.01 (0.001) &         1.00 (0.002) &  1.00 (0.000) \\
   & 200 &  1.01 (0.000) &   1.00 (0.001) &  1.01 (0.000) &         1.00 (0.001) &  1.00 (0.000) \\
   & 500 &  1.01 (0.000) &   1.00 (0.001) &  1.01 (0.001) &         1.00 (0.001) &  \\
 \midrule
20 & 50  &  1.03 (0.002) &   1.01 (0.001) &  1.03 (0.001) &         1.00 (0.002) &  1.01 (0.001) \\
   & 100 &  1.03 (0.001) &   1.00 (0.001) &  1.03 (0.001) &         1.00 (0.000) &  1.01 (0.000) \\
\end{longtable}

\begin{longtable}{lllllll}
\multicolumn{7}{c}{\textbf{Distortions on the \textit{Tower} data set}} \\
\toprule
 \parbox[t]{5mm}{\ \\$k$} & \parbox[t]{5mm}{\ \\$m$} &BICO & \parbox[t]{1.7cm}{Group\\Sampling} &\parbox[t]{1.7cm}{Ray\\Maker}&\parbox[t]{1.7cm}{Sensitivity\\Sampling}&    StreamKM++ \\
\midrule
20  & 50  &  1.18 (0.018) &   1.09 (0.025) &  1.44 (0.045) &         1.04 (0.025) &  1.04 (0.004) \\
    & 100 &  1.11 (0.005) &   1.05 (0.016) &  1.16 (0.013) &         1.03 (0.007) &  1.03 (0.003) \\
    & 200 &  1.06 (0.002) &   1.03 (0.007) &  1.06 (0.007) &         1.02 (0.009) &  1.02 (0.001) \\
    & 500 &  1.03 (0.001) &   1.01 (0.006) &  1.03 (0.002) &         1.01 (0.005) &  \\
 \midrule
40  & 50  &  1.20 (0.012) &   1.11 (0.014) &  1.48 (0.020) &         1.04 (0.019) &  1.05 (0.002) \\
    & 100 &  1.11 (0.006) &   1.05 (0.013) &  1.15 (0.008) &         1.02 (0.010) &  1.03 (0.002) \\
\end{longtable}

\section{Costs}
\label{sec:costs-tables}

\begin{longtable}{lllllll}
\multicolumn{7}{c}{\textbf{Costs on the \textit{Benchmark} data set}} \\
\toprule
\parbox[t]{5mm}{\ \\$k$} & \parbox[t]{5mm}{\ \\$m$} &     BICO &     \parbox[t]{1.7cm}{Group\\Sampling} &     \parbox[t]{1.7cm}{Ray\\Maker}&\parbox[t]{1.7cm}{Sensitivity\\Sampling}&         StreamKM++ \\
\midrule
10 & 50  &  \parbox[t]{17mm}{4.8e+06\\\small(2.3e+05)} &  \parbox[t]{17mm}{5.2e+06\\\small(4.1e+05)} &  \parbox[t]{17mm}{5.0e+06\\\small(2.9e+05)} &  \parbox[t]{17mm}{4.7e+06\\\small(3.8e+05)} &  \parbox[t]{17mm}{5.3e+06\\\small(0.0e+00)} \\
   & 100 &  \parbox[t]{17mm}{4.6e+06\\\small(1.7e+05)} &  \parbox[t]{17mm}{5.3e+06\\\small(1.1e+05)} &  \parbox[t]{17mm}{4.9e+06\\\small(2.8e+05)} &  \parbox[t]{17mm}{4.5e+06\\\small(1.3e+05)} &  \parbox[t]{17mm}{5.3e+06\\\small(0.0e+00)} \\
   & 200 &  \parbox[t]{17mm}{4.7e+06\\\small(3.1e+05)} &  \parbox[t]{17mm}{5.4e+06\\\small(2.4e+05)} &  \parbox[t]{17mm}{5.1e+06\\\small(3.6e+05)} &  \parbox[t]{17mm}{4.6e+06\\\small(2.5e+05)} &  \parbox[t]{17mm}{5.3e+06\\\small(0.0e+00)} \\
   & 500 &  \parbox[t]{17mm}{4.6e+06\\\small(1.0e+05)} &  \parbox[t]{17mm}{5.4e+06\\\small(1.4e+05)} &  \parbox[t]{17mm}{5.1e+06\\\small(2.1e+05)} &  \parbox[t]{17mm}{4.7e+06\\\small(3.6e+05)} &       \\
 \midrule
20 & 50  &  \parbox[t]{17mm}{1.4e+07\\\small(8.8e+05)} &  \parbox[t]{17mm}{1.5e+07\\\small(2.6e+05)} &  \parbox[t]{17mm}{1.3e+07\\\small(6.2e+05)} &  \parbox[t]{17mm}{1.2e+07\\\small(2.9e-09)} &  \parbox[t]{17mm}{1.6e+07\\\small(2.9e+05)} \\
   & 100 &  \parbox[t]{17mm}{1.3e+07\\\small(9.8e+05)} &  \parbox[t]{17mm}{1.5e+07\\\small(4.6e+05)} &  \parbox[t]{17mm}{1.4e+07\\\small(5.9e+05)} &  \parbox[t]{17mm}{1.3e+07\\\small(1.1e+06)} &  \parbox[t]{17mm}{1.6e+07\\\small(1.7e+05)} \\
\end{longtable}

\begin{longtable}{lllllll}
\multicolumn{7}{c}{\textbf{Costs on the \textit{Caltech} data set}} \\
\toprule
\parbox[t]{5mm}{\ \\$k$} & \parbox[t]{5mm}{\ \\$m$} &     BICO &     \parbox[t]{1.7cm}{Group\\Sampling} &     \parbox[t]{1.7cm}{Ray\\Maker}&\parbox[t]{1.7cm}{Sensitivity\\Sampling}&         StreamKM++ \\
\midrule
10 & 50  &  \parbox[t]{17mm}{4.1e+11\\\small(3.2e+09)} &  \parbox[t]{17mm}{4.1e+11\\\small(2.9e+09)} &  \parbox[t]{17mm}{4.0e+11\\\small(1.7e+09)} &  \parbox[t]{17mm}{4.0e+11\\\small(2.5e+09)} &  \parbox[t]{17mm}{4.1e+11\\\small(2.7e+09)} \\
   & 100 &  \parbox[t]{17mm}{4.0e+11\\\small(1.6e+09)} &  \parbox[t]{17mm}{4.0e+11\\\small(3.1e+09)} &  \parbox[t]{17mm}{4.0e+11\\\small(1.7e+09)} &  \parbox[t]{17mm}{4.0e+11\\\small(3.0e+09)} &  \parbox[t]{17mm}{4.0e+11\\\small(8.4e+08)} \\
   & 200 &  \parbox[t]{17mm}{4.0e+11\\\small(3.9e+09)} &  \parbox[t]{17mm}{4.0e+11\\\small(1.5e+09)} &  \parbox[t]{17mm}{4.0e+11\\\small(2.2e+09)} &  \parbox[t]{17mm}{4.0e+11\\\small(1.5e+09)} &  \parbox[t]{17mm}{3.9e+11\\\small(1.1e+09)} \\
   & 500 &  \parbox[t]{17mm}{4.0e+11\\\small(1.8e+09)} &  \parbox[t]{17mm}{3.9e+11\\\small(1.7e+09)} &  \parbox[t]{17mm}{4.0e+11\\\small(1.9e+09)} &  \parbox[t]{17mm}{3.9e+11\\\small(6.3e+08)} &       \\
 \midrule
20 & 50  &  \parbox[t]{17mm}{3.8e+11\\\small(3.1e+09)} &  \parbox[t]{17mm}{3.7e+11\\\small(1.3e+09)} &  \parbox[t]{17mm}{3.7e+11\\\small(2.8e+09)} &  \parbox[t]{17mm}{3.7e+11\\\small(2.4e+09)} &  \parbox[t]{17mm}{3.7e+11\\\small(1.4e+09)} \\
   & 100 &  \parbox[t]{17mm}{3.7e+11\\\small(2.9e+09)} &  \parbox[t]{17mm}{3.7e+11\\\small(1.2e+09)} &  \parbox[t]{17mm}{3.7e+11\\\small(1.9e+09)} &  \parbox[t]{17mm}{3.7e+11\\\small(7.8e+08)} &  \parbox[t]{17mm}{3.7e+11\\\small(1.2e+09)} \\
\end{longtable}

\begin{longtable}{lllllll}
\multicolumn{7}{c}{\textbf{Costs on the \textit{Caltech} data set (with PCA preprocessing)}} \\
\toprule
\parbox[t]{5mm}{\ \\$k$} & \parbox[t]{5mm}{\ \\$m$} &     BICO &     \parbox[t]{1.7cm}{Group\\Sampling} &     \parbox[t]{1.7cm}{Ray\\Maker}&\parbox[t]{1.7cm}{Sensitivity\\Sampling}&         StreamKM++ \\
\midrule
10 & 50  &  \parbox[t]{17mm}{4.1e+11\\\small(1.3e+09)} &  \parbox[t]{17mm}{4.1e+11\\\small(1.1e+09)} &  \parbox[t]{17mm}{4.1e+11\\\small(1.9e+09)} &  \parbox[t]{17mm}{4.1e+11\\\small(1.9e+09)} &  \parbox[t]{17mm}{4.0e+11\\\small(1.5e+09)} \\
   & 100 &  \parbox[t]{17mm}{4.0e+11\\\small(1.0e+09)} &  \parbox[t]{17mm}{4.0e+11\\\small(1.4e+09)} &  \parbox[t]{17mm}{4.0e+11\\\small(1.9e+09)} &  \parbox[t]{17mm}{4.0e+11\\\small(1.2e+09)} &  \parbox[t]{17mm}{4.0e+11\\\small(1.7e+09)} \\
   & 200 &  \parbox[t]{17mm}{4.0e+11\\\small(2.1e+09)} &  \parbox[t]{17mm}{4.0e+11\\\small(1.3e+09)} &  \parbox[t]{17mm}{4.0e+11\\\small(1.8e+09)} &  \parbox[t]{17mm}{4.0e+11\\\small(1.2e+09)} &  \parbox[t]{17mm}{4.0e+11\\\small(8.4e+08)} \\
   & 500 &  \parbox[t]{17mm}{4.0e+11\\\small(1.4e+09)} &  \parbox[t]{17mm}{4.0e+11\\\small(9.7e+08)} &  \parbox[t]{17mm}{4.0e+11\\\small(1.5e+09)} &  \parbox[t]{17mm}{4.0e+11\\\small(9.2e+08)} &       \\
 \midrule
20 & 50  &  \parbox[t]{17mm}{3.7e+11\\\small(1.3e+09)} &  \parbox[t]{17mm}{3.7e+11\\\small(1.3e+09)} &  \parbox[t]{17mm}{3.7e+11\\\small(1.1e+09)} &  \parbox[t]{17mm}{3.7e+11\\\small(1.2e+09)} &  \parbox[t]{17mm}{3.7e+11\\\small(1.6e+09)} \\
   & 100 &  \parbox[t]{17mm}{3.7e+11\\\small(1.3e+09)} &  \parbox[t]{17mm}{3.7e+11\\\small(9.1e+08)} &  \parbox[t]{17mm}{3.7e+11\\\small(9.5e+08)} &  \parbox[t]{17mm}{3.7e+11\\\small(1.0e+09)} &  \parbox[t]{17mm}{3.7e+11\\\small(9.9e+08)} \\
\end{longtable}

\begin{longtable}{lllllll}
\multicolumn{7}{c}{\textbf{Costs on the \textit{Census} data set}} \\
\toprule
\parbox[t]{5mm}{\ \\$k$} & \parbox[t]{5mm}{\ \\$m$} &     BICO &     \parbox[t]{1.7cm}{Group\\Sampling} &     \parbox[t]{1.7cm}{Ray\\Maker}&\parbox[t]{1.7cm}{Sensitivity\\Sampling}&         StreamKM++ \\
\midrule
10 & 50  &  \parbox[t]{17mm}{2.5e+08\\\small(4.7e+06)} &  \parbox[t]{17mm}{2.6e+08\\\small(9.7e+06)} &  \parbox[t]{17mm}{2.5e+08\\\small(7.5e+06)} &  \parbox[t]{17mm}{2.7e+08\\\small(1.5e+07)} &  \parbox[t]{17mm}{2.5e+08\\\small(8.9e+06)} \\
   & 100 &  \parbox[t]{17mm}{2.5e+08\\\small(6.1e+06)} &  \parbox[t]{17mm}{2.6e+08\\\small(1.2e+07)} &  \parbox[t]{17mm}{2.5e+08\\\small(5.1e+06)} &  \parbox[t]{17mm}{2.6e+08\\\small(8.3e+06)} &  \parbox[t]{17mm}{2.6e+08\\\small(7.9e+06)} \\
   & 200 &  \parbox[t]{17mm}{2.5e+08\\\small(5.8e+06)} &  \parbox[t]{17mm}{2.6e+08\\\small(1.2e+07)} &  \parbox[t]{17mm}{2.5e+08\\\small(5.0e+06)} &  \parbox[t]{17mm}{2.7e+08\\\small(1.0e+07)} &  \parbox[t]{17mm}{2.5e+08\\\small(5.1e+06)} \\
   & 500 &  \parbox[t]{17mm}{2.5e+08\\\small(4.3e+06)} &  \parbox[t]{17mm}{2.6e+08\\\small(1.5e+07)} &  \parbox[t]{17mm}{2.4e+08\\\small(2.6e+06)} &  \parbox[t]{17mm}{2.6e+08\\\small(1.2e+07)} &       \\
 \midrule
20 & 50  &  \parbox[t]{17mm}{1.9e+08\\\small(4.3e+06)} &  \parbox[t]{17mm}{2.1e+08\\\small(1.1e+07)} &  \parbox[t]{17mm}{1.9e+08\\\small(2.0e+06)} &  \parbox[t]{17mm}{2.0e+08\\\small(1.1e+07)} &  \parbox[t]{17mm}{1.9e+08\\\small(4.3e+06)} \\
   & 100 &  \parbox[t]{17mm}{1.9e+08\\\small(4.9e+06)} &  \parbox[t]{17mm}{2.0e+08\\\small(7.2e+06)} &  \parbox[t]{17mm}{1.9e+08\\\small(2.1e+06)} &  \parbox[t]{17mm}{2.0e+08\\\small(9.7e+06)} &  \parbox[t]{17mm}{1.9e+08\\\small(4.0e+06)} \\
\end{longtable}

\begin{longtable}{lllllll}
\multicolumn{7}{c}{\textbf{Costs on the \textit{Census} data set (with PCA preprocessing)}} \\
\toprule
\parbox[t]{5mm}{\ \\$k$} & \parbox[t]{5mm}{\ \\$m$} &     BICO &     \parbox[t]{1.7cm}{Group\\Sampling} &     \parbox[t]{1.7cm}{Ray\\Maker}&\parbox[t]{1.7cm}{Sensitivity\\Sampling}&         StreamKM++ \\
\midrule
10 & 50  &  \parbox[t]{17mm}{2.5e+08\\\small(6.2e+06)} &  \parbox[t]{17mm}{2.6e+08\\\small(9.8e+06)} &  \parbox[t]{17mm}{2.5e+08\\\small(4.7e+06)} &  \parbox[t]{17mm}{2.7e+08\\\small(8.1e+06)} &  \parbox[t]{17mm}{2.6e+08\\\small(6.9e+06)} \\
   & 100 &  \parbox[t]{17mm}{2.5e+08\\\small(5.8e+06)} &  \parbox[t]{17mm}{2.7e+08\\\small(8.8e+06)} &  \parbox[t]{17mm}{2.5e+08\\\small(6.8e+06)} &  \parbox[t]{17mm}{2.7e+08\\\small(1.5e+07)} &  \parbox[t]{17mm}{2.6e+08\\\small(1.1e+07)} \\
   & 200 &  \parbox[t]{17mm}{2.5e+08\\\small(4.2e+06)} &  \parbox[t]{17mm}{2.5e+08\\\small(8.2e+06)} &  \parbox[t]{17mm}{2.5e+08\\\small(4.9e+06)} &  \parbox[t]{17mm}{2.7e+08\\\small(1.1e+07)} &  \parbox[t]{17mm}{2.5e+08\\\small(7.2e+06)} \\
   & 500 &  \parbox[t]{17mm}{2.5e+08\\\small(5.9e+06)} &  \parbox[t]{17mm}{2.6e+08\\\small(1.1e+07)} &  \parbox[t]{17mm}{2.5e+08\\\small(7.2e+06)} &  \parbox[t]{17mm}{2.6e+08\\\small(1.2e+07)} &       \\
 \midrule
20 & 50  &  \parbox[t]{17mm}{2.0e+08\\\small(5.2e+06)} &  \parbox[t]{17mm}{2.0e+08\\\small(7.1e+06)} &  \parbox[t]{17mm}{1.9e+08\\\small(5.0e+06)} &  \parbox[t]{17mm}{2.0e+08\\\small(7.6e+06)} &  \parbox[t]{17mm}{2.0e+08\\\small(6.3e+06)} \\
   & 100 &  \parbox[t]{17mm}{1.9e+08\\\small(1.9e+06)} &  \parbox[t]{17mm}{2.0e+08\\\small(5.3e+06)} &  \parbox[t]{17mm}{1.9e+08\\\small(4.1e+06)} &  \parbox[t]{17mm}{2.1e+08\\\small(1.4e+07)} &  \parbox[t]{17mm}{1.9e+08\\\small(3.4e+06)} \\
\end{longtable}

\begin{longtable}{lllllll}
\multicolumn{7}{c}{\textbf{Costs on the \textit{Covertype} data set}} \\
\toprule
\parbox[t]{5mm}{\ \\$k$} & \parbox[t]{5mm}{\ \\$m$} &     BICO &     \parbox[t]{1.7cm}{Group\\Sampling} &     \parbox[t]{1.7cm}{Ray\\Maker}&\parbox[t]{1.7cm}{Sensitivity\\Sampling}&         StreamKM++ \\
\midrule
10 & 50  &  \parbox[t]{17mm}{3.5e+11\\\small(5.4e+09)} &  \parbox[t]{17mm}{3.6e+11\\\small(4.8e+09)} &  \parbox[t]{17mm}{3.5e+11\\\small(5.5e+09)} &  \parbox[t]{17mm}{3.6e+11\\\small(9.1e+09)} &  \parbox[t]{17mm}{3.5e+11\\\small(8.5e+09)} \\
   & 100 &  \parbox[t]{17mm}{3.5e+11\\\small(4.4e+09)} &  \parbox[t]{17mm}{3.5e+11\\\small(4.8e+09)} &  \parbox[t]{17mm}{3.5e+11\\\small(6.2e+09)} &  \parbox[t]{17mm}{3.6e+11\\\small(4.9e+09)} &  \parbox[t]{17mm}{3.5e+11\\\small(5.6e+09)} \\
   & 200 &  \parbox[t]{17mm}{3.4e+11\\\small(3.2e+09)} &  \parbox[t]{17mm}{3.5e+11\\\small(5.2e+09)} &  \parbox[t]{17mm}{3.5e+11\\\small(5.1e+09)} &  \parbox[t]{17mm}{3.5e+11\\\small(5.7e+09)} &  \parbox[t]{17mm}{3.4e+11\\\small(1.8e+09)} \\
   & 500 &  \parbox[t]{17mm}{3.4e+11\\\small(3.3e+09)} &  \parbox[t]{17mm}{3.4e+11\\\small(5.1e+09)} &  \parbox[t]{17mm}{3.5e+11\\\small(6.7e+09)} &  \parbox[t]{17mm}{3.5e+11\\\small(3.2e+09)} &       \\
 \midrule
20 & 50  &  \parbox[t]{17mm}{2.1e+11\\\small(2.5e+09)} &  \parbox[t]{17mm}{2.1e+11\\\small(2.3e+09)} &  \parbox[t]{17mm}{2.1e+11\\\small(1.9e+09)} &  \parbox[t]{17mm}{2.2e+11\\\small(4.1e+09)} &  \parbox[t]{17mm}{2.1e+11\\\small(2.8e+09)} \\
   & 100 &  \parbox[t]{17mm}{2.1e+11\\\small(1.9e+09)} &  \parbox[t]{17mm}{2.1e+11\\\small(3.0e+09)} &  \parbox[t]{17mm}{2.1e+11\\\small(2.4e+09)} &  \parbox[t]{17mm}{2.1e+11\\\small(4.2e+09)} &  \parbox[t]{17mm}{2.1e+11\\\small(1.5e+09)} \\
\end{longtable}

\begin{longtable}{lllllll}
\multicolumn{7}{c}{\textbf{Costs on the \textit{Covertype} data set (with PCA preprocessing)}} \\
\toprule
\parbox[t]{5mm}{\ \\$k$} & \parbox[t]{5mm}{\ \\$m$} &     BICO &     \parbox[t]{1.7cm}{Group\\Sampling} &     \parbox[t]{1.7cm}{Ray\\Maker}&\parbox[t]{1.7cm}{Sensitivity\\Sampling}&         StreamKM++ \\
\midrule
10 & 50  &  \parbox[t]{17mm}{3.5e+11\\\small(5.4e+09)} &  \parbox[t]{17mm}{3.6e+11\\\small(1.0e+10)} &  \parbox[t]{17mm}{3.5e+11\\\small(7.4e+09)} &  \parbox[t]{17mm}{3.6e+11\\\small(9.2e+09)} &  \parbox[t]{17mm}{3.5e+11\\\small(5.9e+09)} \\
   & 100 &  \parbox[t]{17mm}{3.5e+11\\\small(4.6e+09)} &  \parbox[t]{17mm}{3.5e+11\\\small(6.8e+09)} &  \parbox[t]{17mm}{3.5e+11\\\small(3.2e+09)} &  \parbox[t]{17mm}{3.6e+11\\\small(3.5e+09)} &  \parbox[t]{17mm}{3.5e+11\\\small(3.0e+09)} \\
   & 200 &  \parbox[t]{17mm}{3.4e+11\\\small(3.1e+09)} &  \parbox[t]{17mm}{3.5e+11\\\small(6.9e+09)} &  \parbox[t]{17mm}{3.5e+11\\\small(6.7e+09)} &  \parbox[t]{17mm}{3.5e+11\\\small(3.6e+09)} &  \parbox[t]{17mm}{3.4e+11\\\small(3.8e+08)} \\
   & 500 &  \parbox[t]{17mm}{3.4e+11\\\small(2.2e+09)} &  \parbox[t]{17mm}{3.4e+11\\\small(3.5e+09)} &  \parbox[t]{17mm}{3.5e+11\\\small(3.1e+09)} &  \parbox[t]{17mm}{3.5e+11\\\small(3.7e+09)} &       \\
 \midrule
20 & 50  &  \parbox[t]{17mm}{2.1e+11\\\small(1.7e+09)} &  \parbox[t]{17mm}{2.1e+11\\\small(2.6e+09)} &  \parbox[t]{17mm}{2.1e+11\\\small(2.7e+09)} &  \parbox[t]{17mm}{2.2e+11\\\small(4.7e+09)} &  \parbox[t]{17mm}{2.1e+11\\\small(4.3e+09)} \\
   & 100 &  \parbox[t]{17mm}{2.1e+11\\\small(2.6e+09)} &  \parbox[t]{17mm}{2.1e+11\\\small(3.7e+09)} &  \parbox[t]{17mm}{2.1e+11\\\small(2.5e+09)} &  \parbox[t]{17mm}{2.1e+11\\\small(1.6e+09)} &  \parbox[t]{17mm}{2.1e+11\\\small(1.6e+09)} \\
\end{longtable}

\begin{longtable}{lllllll}
\multicolumn{7}{c}{\textbf{Costs on the \textit{NYTimes} data set}} \\
\toprule
\parbox[t]{5mm}{\ \\$k$} & \parbox[t]{5mm}{\ \\$m$} &     BICO &     \parbox[t]{1.7cm}{Group\\Sampling} &     \parbox[t]{1.7cm}{Ray\\Maker}&\parbox[t]{1.7cm}{Sensitivity\\Sampling}&         StreamKM++ \\
\midrule
10 & 50  &  \parbox[t]{17mm}{2.8e+08\\\small(7.7e+05)} &  \parbox[t]{17mm}{2.8e+08\\\small(4.6e+05)} &  \parbox[t]{17mm}{2.8e+08\\\small(3.6e+05)} &  \parbox[t]{17mm}{2.8e+08\\\small(5.2e+05)} &  \parbox[t]{17mm}{2.9e+08\\\small(1.5e+06)} \\
   & 100 &  \parbox[t]{17mm}{2.8e+08\\\small(7.7e+05)} &  \parbox[t]{17mm}{2.8e+08\\\small(4.4e+05)} &  \parbox[t]{17mm}{2.8e+08\\\small(4.9e+05)} &  \parbox[t]{17mm}{2.8e+08\\\small(5.1e+05)} &  \parbox[t]{17mm}{2.8e+08\\\small(7.2e+05)} \\
   & 200 &  \parbox[t]{17mm}{2.8e+08\\\small(9.3e+05)} &  \parbox[t]{17mm}{2.8e+08\\\small(8.6e+05)} &  \parbox[t]{17mm}{2.8e+08\\\small(7.7e+05)} &  \parbox[t]{17mm}{2.8e+08\\\small(5.2e+05)} &  \parbox[t]{17mm}{2.8e+08\\\small(7.4e+05)} \\
   & 500 &  \parbox[t]{17mm}{2.8e+08\\\small(4.5e+05)} &  \parbox[t]{17mm}{2.8e+08\\\small(3.0e+05)} &  \parbox[t]{17mm}{2.8e+08\\\small(5.6e+05)} &  \parbox[t]{17mm}{2.8e+08\\\small(3.9e+05)} &       \\
 \midrule
20 & 50  &  \parbox[t]{17mm}{2.8e+08\\\small(9.7e+05)} &  \parbox[t]{17mm}{2.8e+08\\\small(1.2e+06)} &  \parbox[t]{17mm}{2.8e+08\\\small(5.5e+05)} &  \parbox[t]{17mm}{2.8e+08\\\small(9.1e+05)} &  \parbox[t]{17mm}{2.8e+08\\\small(5.3e+05)} \\
   & 100 &  \parbox[t]{17mm}{2.7e+08\\\small(9.1e+05)} &  \parbox[t]{17mm}{2.8e+08\\\small(6.8e+05)} &  \parbox[t]{17mm}{2.7e+08\\\small(5.1e+05)} &  \parbox[t]{17mm}{2.8e+08\\\small(8.8e+05)} &  \parbox[t]{17mm}{2.8e+08\\\small(7.5e+05)} \\
\end{longtable}

\begin{longtable}{lllllll}
\multicolumn{7}{c}{\textbf{Costs on the \textit{NYTimes} data set (with PCA preprocessing)}} \\
\toprule
\parbox[t]{5mm}{\ \\$k$} & \parbox[t]{5mm}{\ \\$m$} &     BICO &     \parbox[t]{1.7cm}{Group\\Sampling} &     \parbox[t]{1.7cm}{Ray\\Maker}&\parbox[t]{1.7cm}{Sensitivity\\Sampling}&         StreamKM++ \\
\midrule
10 & 50  &  \parbox[t]{17mm}{2.7e+08\\\small(5.5e+05)} &  \parbox[t]{17mm}{2.7e+08\\\small(5.2e+05)} &  \parbox[t]{17mm}{2.7e+08\\\small(4.7e+05)} &  \parbox[t]{17mm}{2.7e+08\\\small(5.3e+05)} &  \parbox[t]{17mm}{2.7e+08\\\small(5.4e+05)} \\
   & 100 &  \parbox[t]{17mm}{2.7e+08\\\small(4.5e+05)} &  \parbox[t]{17mm}{2.7e+08\\\small(4.5e+05)} &  \parbox[t]{17mm}{2.7e+08\\\small(6.9e+05)} &  \parbox[t]{17mm}{2.7e+08\\\small(7.4e+05)} &  \parbox[t]{17mm}{2.7e+08\\\small(3.4e+05)} \\
   & 200 &  \parbox[t]{17mm}{2.7e+08\\\small(4.9e+05)} &  \parbox[t]{17mm}{2.7e+08\\\small(3.4e+05)} &  \parbox[t]{17mm}{2.7e+08\\\small(5.0e+05)} &  \parbox[t]{17mm}{2.7e+08\\\small(5.9e+05)} &  \parbox[t]{17mm}{2.7e+08\\\small(8.4e+04)} \\
   & 500 &  \parbox[t]{17mm}{2.7e+08\\\small(2.2e+05)} &  \parbox[t]{17mm}{2.7e+08\\\small(2.5e+05)} &  \parbox[t]{17mm}{2.7e+08\\\small(3.6e+05)} &  \parbox[t]{17mm}{2.7e+08\\\small(5.1e+05)} &       \\
 \midrule
20 & 50  &  \parbox[t]{17mm}{2.6e+08\\\small(4.2e+05)} &  \parbox[t]{17mm}{2.6e+08\\\small(4.3e+05)} &  \parbox[t]{17mm}{2.6e+08\\\small(2.6e+05)} &  \parbox[t]{17mm}{2.6e+08\\\small(4.4e+05)} &  \parbox[t]{17mm}{2.6e+08\\\small(2.5e+05)} \\
   & 100 &  \parbox[t]{17mm}{2.6e+08\\\small(3.8e+05)} &  \parbox[t]{17mm}{2.6e+08\\\small(5.9e+05)} &  \parbox[t]{17mm}{2.6e+08\\\small(3.2e+05)} &  \parbox[t]{17mm}{2.6e+08\\\small(5.1e+05)} &  \parbox[t]{17mm}{2.6e+08\\\small(2.5e+05)} \\
\end{longtable}

\begin{longtable}{lllllll}
\multicolumn{7}{c}{\textbf{Costs on the \textit{Tower} data set}} \\
\toprule
 \parbox[t]{5mm}{\ \\$k$} & \parbox[t]{5mm}{\ \\$m$} &     BICO &     \parbox[t]{1.7cm}{Group\\Sampling} &     \parbox[t]{1.7cm}{Ray\\Maker}&\parbox[t]{1.7cm}{Sensitivity\\Sampling}&         StreamKM++ \\
\midrule
20  & 50  &  \parbox[t]{17mm}{6.5e+08\\\small(1.3e+07)} &  \parbox[t]{17mm}{6.6e+08\\\small(3.4e+07)} &  \parbox[t]{17mm}{6.6e+08\\\small(1.8e+07)} &  \parbox[t]{17mm}{6.6e+08\\\small(1.4e+07)} &  \parbox[t]{17mm}{6.5e+08\\\small(2.3e+07)} \\
    & 100 &  \parbox[t]{17mm}{6.4e+08\\\small(1.8e+07)} &  \parbox[t]{17mm}{6.5e+08\\\small(2.3e+07)} &  \parbox[t]{17mm}{6.3e+08\\\small(1.3e+07)} &  \parbox[t]{17mm}{6.8e+08\\\small(4.2e+07)} &  \parbox[t]{17mm}{6.4e+08\\\small(1.4e+07)} \\
    & 200 &  \parbox[t]{17mm}{6.3e+08\\\small(4.8e+06)} &  \parbox[t]{17mm}{6.4e+08\\\small(2.0e+07)} &  \parbox[t]{17mm}{6.2e+08\\\small(7.0e+06)} &  \parbox[t]{17mm}{6.5e+08\\\small(1.9e+07)} &  \parbox[t]{17mm}{6.3e+08\\\small(1.4e+07)} \\
    & 500 &  \parbox[t]{17mm}{6.2e+08\\\small(8.4e+06)} &  \parbox[t]{17mm}{6.5e+08\\\small(4.8e+07)} &  \parbox[t]{17mm}{6.2e+08\\\small(1.2e+07)} &  \parbox[t]{17mm}{6.4e+08\\\small(2.0e+07)} &       \\
 \midrule
40  & 50  &  \parbox[t]{17mm}{3.4e+08\\\small(7.3e+06)} &  \parbox[t]{17mm}{3.6e+08\\\small(6.3e+06)} &  \parbox[t]{17mm}{3.4e+08\\\small(5.2e+06)} &  \parbox[t]{17mm}{3.6e+08\\\small(8.4e+06)} &  \parbox[t]{17mm}{3.4e+08\\\small(4.9e+06)} \\
    & 100 &  \parbox[t]{17mm}{3.4e+08\\\small(9.1e+06)} &  \parbox[t]{17mm}{3.5e+08\\\small(5.9e+06)} &  \parbox[t]{17mm}{3.4e+08\\\small(6.0e+06)} &  \parbox[t]{17mm}{3.5e+08\\\small(7.9e+06)} &  \parbox[t]{17mm}{3.4e+08\\\small(1.6e+06)} \\
\end{longtable}

\section{Running Times}
\label{sec:running-times-tables}

\begin{longtable}{llrrrrr}
\multicolumn{7}{c}{\textbf{Running times on the \textit{Benchmark} data set}} \\
\toprule
\parbox[t]{5mm}{\ \\$k$} & \parbox[t]{5mm}{\ \\$m$} &     BICO &  \parbox[t]{1.5cm}{Group\\Sampling}  &   Ray Maker & \parbox[t]{1.5cm}{Sensitivity\\Sampling} &    StreamKM++ \\
\midrule
10 & 50  &    19 (3) &       117 (10) &    138 (15) &             106 (10) &      736 (60) \\
   & 100 &    20 (4) &        120 (7) &    143 (17) &              107 (8) &     1460 (87) \\
   & 200 &    25 (3) &       118 (10) &    151 (19) &              108 (8) &    2889 (139) \\
   & 500 &    31 (7) &        125 (9) &    188 (10) &             105 (14) &  \\
 \midrule
20 & 50  &   104 (7) &      1016 (74) &    904 (68) &             933 (81) &    7233 (555) \\
   & 100 &  124 (14) &      1021 (77) &    918 (72) &            913 (112) &   14546 (528) \\
\end{longtable}

\begin{longtable}{llrrrrr}
\multicolumn{7}{c}{\textbf{Running times on the \textit{Caltech} data set}} \\
\toprule
\parbox[t]{5mm}{\ \\$k$} & \parbox[t]{5mm}{\ \\$m$} &     BICO &  \parbox[t]{1.5cm}{Group\\Sampling}  &   Ray Maker & \parbox[t]{1.5cm}{Sensitivity\\Sampling} &     StreamKM++ \\
\midrule
10 & 50  &   133 (9) &      1039 (33) &  1146 (134) &             940 (98) &     5799 (309) \\
   & 100 &  151 (12) &      1020 (55) &   1202 (49) &             949 (92) &    11026 (675) \\
   & 200 &  180 (22) &      1029 (53) &  1218 (118) &             974 (34) &    21922 (952) \\
   & 500 &  233 (17) &       933 (65) &   1302 (39) &             877 (61) &   \\
 \midrule
20 & 50  &  144 (16) &     1549 (113) &   1400 (85) &           1447 (106) &    11418 (441) \\
   & 100 &  191 (18) &      1566 (72) &   1439 (37) &            1458 (90) &    22963 (649) \\
\end{longtable}

\begin{longtable}{llrrrrr}
\multicolumn{7}{c}{\textbf{Running times on the \textit{Caltech} data set (with PCA preprocessing)}} \\
\toprule
\parbox[t]{5mm}{\ \\$k$} & \parbox[t]{5mm}{\ \\$m$} &     BICO &  \parbox[t]{1.5cm}{Group\\Sampling}  &   Ray Maker & \parbox[t]{1.5cm}{Sensitivity\\Sampling} &    StreamKM++ \\
\midrule
10 & 50  &  325 (16) &     1452 (112) &  1569 (174) &           1357 (153) &    6156 (268) \\
   & 100 &  326 (25) &     1449 (118) &  1596 (135) &           1360 (114) &   11401 (454) \\
   & 200 &  334 (27) &     1425 (161) &  1568 (222) &           1383 (113) &  21289 (1324) \\
   & 500 &  386 (27) &      1386 (91) &   1685 (94) &            1357 (77) &  \\
 \midrule
20 & 50  &  362 (26) &     1934 (179) &  1753 (166) &           1783 (223) &   10751 (914) \\
   & 100 &  379 (29) &     1895 (222) &  1796 (128) &           1833 (143) &  20680 (1273) \\
\end{longtable}

\begin{longtable}{llrrrrr}
\multicolumn{7}{c}{\textbf{Running times on the \textit{Census} data set}} \\
\toprule
\parbox[t]{5mm}{\ \\$k$} & \parbox[t]{5mm}{\ \\$m$} &     BICO &  \parbox[t]{1.5cm}{Group\\Sampling}  & Ray Maker & \parbox[t]{1.5cm}{Sensitivity\\Sampling} &    StreamKM++ \\
\midrule
10 & 50  &    47 (6) &       265 (33) &  300 (34) &             239 (37) &    1866 (176) \\
   & 100 &    57 (2) &       318 (27) &  400 (15) &             291 (23) &    4080 (273) \\
   & 200 &    61 (4) &       329 (18) &  438 (26) &             286 (29) &    8203 (231) \\
   & 500 &    69 (4) &       314 (35) &  534 (36) &             281 (42) &  \\
 \midrule
20 & 50  &    48 (5) &       444 (61) &  373 (49) &             401 (54) &    3584 (239) \\
   & 100 &    55 (5) &       512 (40) &  474 (28) &             469 (12) &    8107 (529) \\
\end{longtable}

\begin{longtable}{llrrrrr}
\multicolumn{7}{c}{\textbf{Running times on the \textit{Census} data set (with PCA preprocessing)}} \\
\toprule
\parbox[t]{5mm}{\ \\$k$} & \parbox[t]{5mm}{\ \\$m$} &     BICO &  \parbox[t]{1.5cm}{Group\\Sampling}  &  Ray Maker & \parbox[t]{1.5cm}{Sensitivity\\Sampling} &    StreamKM++ \\
\midrule
10 & 50  &   103 (8) &       529 (38) &   535 (87) &             479 (40) &    2089 (224) \\
   & 100 &   105 (8) &       505 (61) &   593 (49) &             481 (37) &    4153 (376) \\
   & 200 &   106 (6) &       517 (40) &   630 (62) &             477 (45) &    7587 (443) \\
   & 500 &  114 (10) &       483 (38) &   706 (54) &             463 (28) &  \\
 \midrule
20 & 50  &   107 (7) &       702 (62) &   646 (48) &             635 (74) &    4037 (384) \\
   & 100 &   110 (7) &       694 (71) &   667 (46) &             642 (58) &    8010 (410) \\
\end{longtable}

\begin{longtable}{llrrrrr}
\multicolumn{7}{c}{\textbf{Running times on the \textit{Covertype} data set}} \\
\toprule
\parbox[t]{5mm}{\ \\$k$} & \parbox[t]{5mm}{\ \\$m$} &   BICO &  \parbox[t]{1.5cm}{Group\\Sampling}  & Ray Maker & \parbox[t]{1.5cm}{Sensitivity\\Sampling} &  StreamKM++ \\
\midrule
10 & 50  &   8 (1) &         58 (9) &    60 (8) &               50 (8) &    275 (37) \\
   & 100 &   8 (1) &         61 (7) &    67 (9) &               48 (9) &    597 (69) \\
   & 200 &   9 (1) &         57 (7) &    85 (9) &               51 (9) &  1322 (214) \\
   & 500 &  12 (1) &         67 (7) &  123 (18) &               55 (5) &\\
 \midrule
20 & 50  &   9 (1) &        93 (13) &   70 (10) &              77 (12) &    580 (80) \\
   & 100 &   9 (1) &         88 (8) &   81 (11) &              74 (12) &  1038 (183) \\
\end{longtable}

\begin{longtable}{llrrrrr}
\multicolumn{7}{c}{\textbf{Running times on the \textit{Covertype} data set (with PCA preprocessing)}} \\
\toprule
\parbox[t]{5mm}{\ \\$k$} & \parbox[t]{5mm}{\ \\$m$} &   BICO &  \parbox[t]{1.5cm}{Group\\Sampling}  & Ray Maker & \parbox[t]{1.5cm}{Sensitivity\\Sampling} &  StreamKM++ \\
\midrule
10 & 50  &  19 (2) &        100 (8) &  101 (10) &               86 (8) &    417 (45) \\
   & 100 &  18 (1) &         99 (7) &   111 (9) &              84 (10) &    797 (58) \\
   & 200 &  19 (1) &         98 (7) &  118 (14) &               87 (8) &   1579 (63) \\
   & 500 &  22 (1) &        90 (10) &  148 (14) &               84 (7) &\\
 \midrule
20 & 50  &  19 (1) &       141 (10) &  116 (10) &              119 (9) &    781 (91) \\
   & 100 &  19 (2) &       137 (13) &   122 (8) &             116 (11) &  1452 (186) \\
\end{longtable}

\begin{longtable}{llrrrrr}
\multicolumn{7}{c}{\textbf{Running times on the \textit{NYTimes} data set}} \\
\toprule
\parbox[t]{5mm}{\ \\$k$} & \parbox[t]{5mm}{\ \\$m$} &       BICO &  \parbox[t]{1.5cm}{Group\\Sampling}  & Ray Maker & \parbox[t]{1.5cm}{Sensitivity\\Sampling} &  StreamKM++ \\
\midrule
10 & 50  &      12 (2) &         58 (6) &    68 (3) &               56 (1) &    369 (19) \\
   & 100 &      16 (5) &         54 (2) &    65 (4) &               49 (3) &    708 (42) \\
   & 200 &      33 (7) &         51 (2) &    64 (2) &               49 (3) &  1373 (173) \\
   & 500 &    107 (24) &         52 (6) &    76 (8) &               52 (4) &\\
 \midrule
20 & 50  &      20 (5) &         92 (5) &    84 (2) &               88 (3) &    670 (70) \\
   & 100 &     36 (12) &         85 (6) &    75 (5) &               79 (6) &  1338 (112) \\
\end{longtable}

\begin{longtable}{llrrrrr}
\multicolumn{7}{c}{\textbf{Running times on the \textit{NYTimes} data set (with PCA preprocessing)}} \\
\toprule
\parbox[t]{5mm}{\ \\$k$} & \parbox[t]{5mm}{\ \\$m$} &     BICO &  \parbox[t]{1.5cm}{Group\\Sampling}  & Ray Maker & \parbox[t]{1.5cm}{Sensitivity\\Sampling} &  StreamKM++ \\
\midrule
10 & 50  &    18 (0) &         28 (1) &    28 (0) &               26 (0) &      71 (3) \\
   & 100 &    18 (0) &         28 (2) &    30 (1) &               25 (1) &    109 (17) \\
   & 200 &    19 (0) &         28 (1) &    33 (1) &               26 (0) &    199 (19) \\
   & 500 &    20 (0) &         29 (0) &    47 (4) &               26 (1) &\\
 \midrule
20 & 50  &    30 (0) &         51 (4) &    50 (1) &               49 (2) &     203 (3) \\
   & 100 &    31 (0) &         54 (2) &    50 (3) &               50 (2) &     375 (6) \\
\end{longtable}

\begin{longtable}{llrrrrr}
\multicolumn{7}{c}{\textbf{Running times on the \textit{Tower} data set}} \\
\toprule
 \parbox[t]{5mm}{\ \\$k$} & \parbox[t]{5mm}{\ \\$m$} &   BICO &  \parbox[t]{1.5cm}{Group\\Sampling}  & Ray Maker & \parbox[t]{1.5cm}{Sensitivity\\Sampling} &   StreamKM++ \\
\midrule
20  & 50  &   8 (2) &       266 (12) &    66 (3) &               68 (2) &     685 (11) \\
    & 100 &   8 (2) &       268 (13) &    91 (5) &               71 (4) &    1373 (14) \\
    & 200 &   7 (1) &       268 (12) &  135 (12) &               71 (4) &    2746 (14) \\
    & 500 &   8 (2) &        265 (8) &  268 (11) &               70 (2) & \\
 \midrule
40  & 50  &   8 (2) &       546 (10) &    92 (2) &              133 (5) &    1368 (21) \\
    & 100 &   7 (1) &       552 (14) &   112 (2) &              134 (4) &    2747 (21) \\
\end{longtable}

\section{Generating Candidate Solutions}
\label{sec:candidate-solution-generation}
As mentioned in~\cref{sec:evaluation-procedure}, it is difficult to construct a candidate solution consisting of $k$ centers which yields a maximal distortion. We experimented with three methods for candidate solution generation; find $k$ centers via $k$-means++, randomly generate $k$ points inside the convex hull (Random CH), and randomly generate $k$ points inside the minimum enclosing ball (Random MEB). Our experiments show that the $k$-means++ method generates candidate solutions which consistently result in large distortions across all algorithms and data sets. The numerical results are presented in~\cref{tab:comparison-solution-generation-all}. Notice that the the numbers in~\cref{tab:comparison-solution-generation-all} are the maximum distortions across $k$ and $m$ because
we are only interesting in determining which of the three methods is superior. \cref{tab:comparison-solution-generation-caltech-bico} shows all the numbers for BICO on the \textit{Caltech} data set. For all the other combinations of algorithms and data sets, the takeaway message is virtually the same; $k$-means++ is the best method for generating candidate solutions.

\begin{longtable}{llrrr}
\toprule
 Data set  & Algorithm     &        $k$-means++&        Random CH&       Random MEB\\
\midrule
Caltech & BICO &   7.499068 &   4.800188 &   2.629218 \\
      & Group Sampling &   1.068441 &   1.035997 &   1.006828 \\
      & Ray Maker &   7.053856 &   5.442665 &   3.453930 \\
      & Sensitivity Sampling &   1.048977 &   1.016564 &   1.029646 \\
      & StreamKM++ &   1.202426 &   1.111698 &   1.012216 \\
\midrule
Caltech+PCA & BICO &   2.862242 &   2.411732 &   1.825558 \\
      & Group Sampling &   1.058707 &   1.021386 &   1.007735 \\
      & Ray Maker &   2.941248 &   2.563131 &   2.053860 \\
      & Sensitivity Sampling &   1.040106 &   1.012539 &   1.032113 \\
      & StreamKM++ &   1.139141 &   1.073384 &   1.013089 \\
\midrule
Census & BICO &   2.502564 &   1.441889 &   1.040356 \\
      & Group Sampling &   1.077525 &   1.032349 &   1.028025 \\
      & Ray Maker &   2.483973 &   1.633945 &   1.236878 \\
      & Sensitivity Sampling &   1.032978 &   1.088607 &   1.079925 \\
      & StreamKM++ &   1.127443 &   1.044656 &   1.003459 \\
\midrule
Census+PCA & BICO &   2.454087 &   1.472404 &   1.030620 \\
      & Group Sampling &   1.070332 &   1.036615 &   1.031606 \\
      & Ray Maker &   2.456844 &   1.617594 &   1.234714 \\
      & Sensitivity Sampling &   1.036155 &   1.086034 &   1.083994 \\
      & StreamKM++ &   1.126928 &   1.042801 &   1.003570 \\
\midrule
Covertype & BICO &   1.294620 &   1.146309 &   1.037080 \\
      & Group Sampling &   1.094779 &   1.041165 &   1.023978 \\
      & Ray Maker &   1.388141 &   1.227148 &   1.121157 \\
      & Sensitivity Sampling &   1.044575 &   1.115355 &   1.103073 \\
      & StreamKM++ &   1.048912 &   1.019965 &   1.004857 \\
\midrule
Covertype+PCA & BICO &   1.292637 &   1.147499 &   1.037594 \\
      & Group Sampling &   1.097391 &   1.038144 &   1.023664 \\
      & Ray Maker &   1.385070 &   1.229436 &   1.122238 \\
      & Sensitivity Sampling &   1.039479 &   1.097699 &   1.092900 \\
      & StreamKM++ &   1.049769 &   1.022796 &   1.006932 \\
\midrule
NYTimes & BICO &  35.319657 &  31.015943 &  28.932553 \\
      & Group Sampling &   1.089067 &   1.061821 &   1.033067 \\
      & Ray Maker &  28.373678 &  24.930933 &  19.520125 \\
      & Sensitivity Sampling &   1.050988 &   1.027549 &   1.012441 \\
      & StreamKM++ &   2.039047 &   1.861497 &   1.674643 \\
\midrule
Tower & BICO &   1.199658 &   1.040606 &   1.005894 \\
      & Group Sampling &   1.113305 &   1.043449 &   1.030855 \\
      & Ray Maker &   1.490978 &   1.286890 &   1.234467 \\
      & Sensitivity Sampling &   1.044464 &   1.103949 &   1.075281 \\
      & StreamKM++ &   1.051491 &   1.009577 &   1.000630 \\
\bottomrule
\caption{Distortions across experiments on different algorithms and data sets.}
\label{tab:comparison-solution-generation-all}
\end{longtable}

\begin{longtable}{rrlll}
\multicolumn{5}{c}{\textbf{Distortions of BICO on the \textit{Caltech} data set}} \\
\toprule
 $k$ &   $m$ &    $k$-means++ &    Random CH &   Random MEB \\
\midrule
10  &  50 & 5.12 (0.307) & 4.30 (0.270) & 2.63 (0.166) \\
    & 100 & 4.48 (0.284) & 3.91 (0.240) & 2.54 (0.125) \\
    & 200 & 4.08 (0.319) & 3.54 (0.236) & 2.46 (0.134) \\
    & 500 & 3.41 (0.215) & 3.06 (0.155) & 2.18 (0.080) \\
\midrule
20  &  50 & 6.35 (1.173) & 4.73 (0.632) & 2.63 (0.147) \\
    & 100 & 4.65 (0.283) & 3.82 (0.245) & 2.41 (0.104) \\
    & 200 & 4.19 (0.384) & 3.49 (0.271) & 2.26 (0.141) \\
    & 500 & 3.50 (0.404) & 3.07 (0.301) & 2.10 (0.128) \\
\midrule
30  &  50 & 6.01 (0.335) & 4.32 (0.204) & 2.57 (0.096) \\
    & 100 & 5.10 (0.628) & 3.89 (0.340) & 2.35 (0.125) \\
    & 200 & 4.29 (0.659) & 3.47 (0.440) & 2.21 (0.170) \\
    & 500 & 3.09 (0.138) & 2.69 (0.122) & 1.87 (0.064) \\
\midrule
40  &  50 & 6.24 (0.524) & 4.33 (0.228) & 2.44 (0.076) \\
    & 100 & 5.23 (0.874) & 3.86 (0.434) & 2.29 (0.191) \\
    & 200 & 4.50 (1.085) & 3.49 (0.595) & 2.16 (0.183) \\
    & 500 & 3.38 (0.398) & 2.85 (0.258) & 1.92 (0.111) \\
\midrule
50  &  50 & 7.50 (1.013) & 4.80 (0.449) & 2.47 (0.190) \\
    & 100 & 5.21 (0.968) & 3.76 (0.475) & 2.20 (0.149) \\
    & 200 & 4.21 (0.296) & 3.35 (0.195) & 2.10 (0.065) \\
    & 500 & 3.36 (0.429) & 2.81 (0.288) & 1.86 (0.105) \\
\bottomrule
\caption{The effect of different solution generation approaches on the distortions (see \cref{sec:candidate-solution-generation}).
The distortions are obtained by running BICO on the \textit{Caltech} data set.
}
\label{tab:comparison-solution-generation-caltech-bico}
\end{longtable}

\newpage
\section{Auxiliary Plots}

\begin{figure}[H]
  \includegraphics[width=0.9\linewidth]{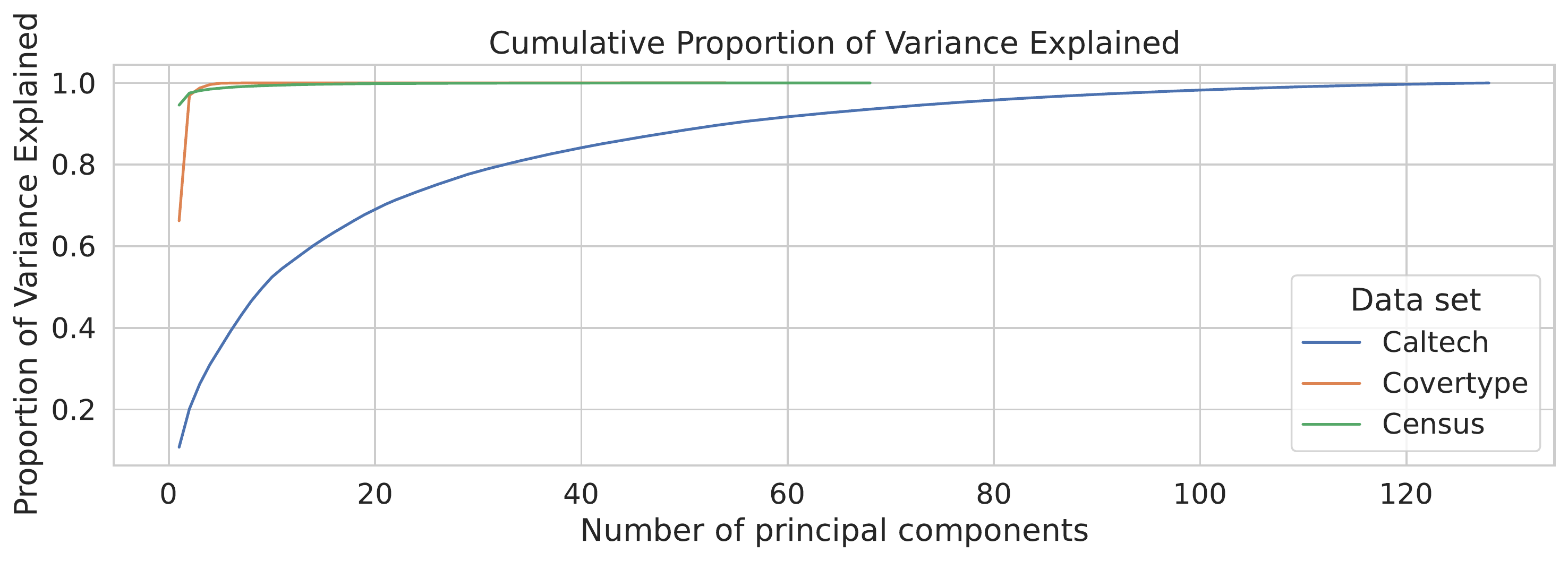}
  \caption{The cumulative proportion of explained variance by principal components on \textit{Caltech}, \textit{Covertype}, and \textit{Census}.}
  \label{fig:explained-variance-pca}
\end{figure}

\begin{figure}[H]
  \includegraphics[width=0.9\linewidth]{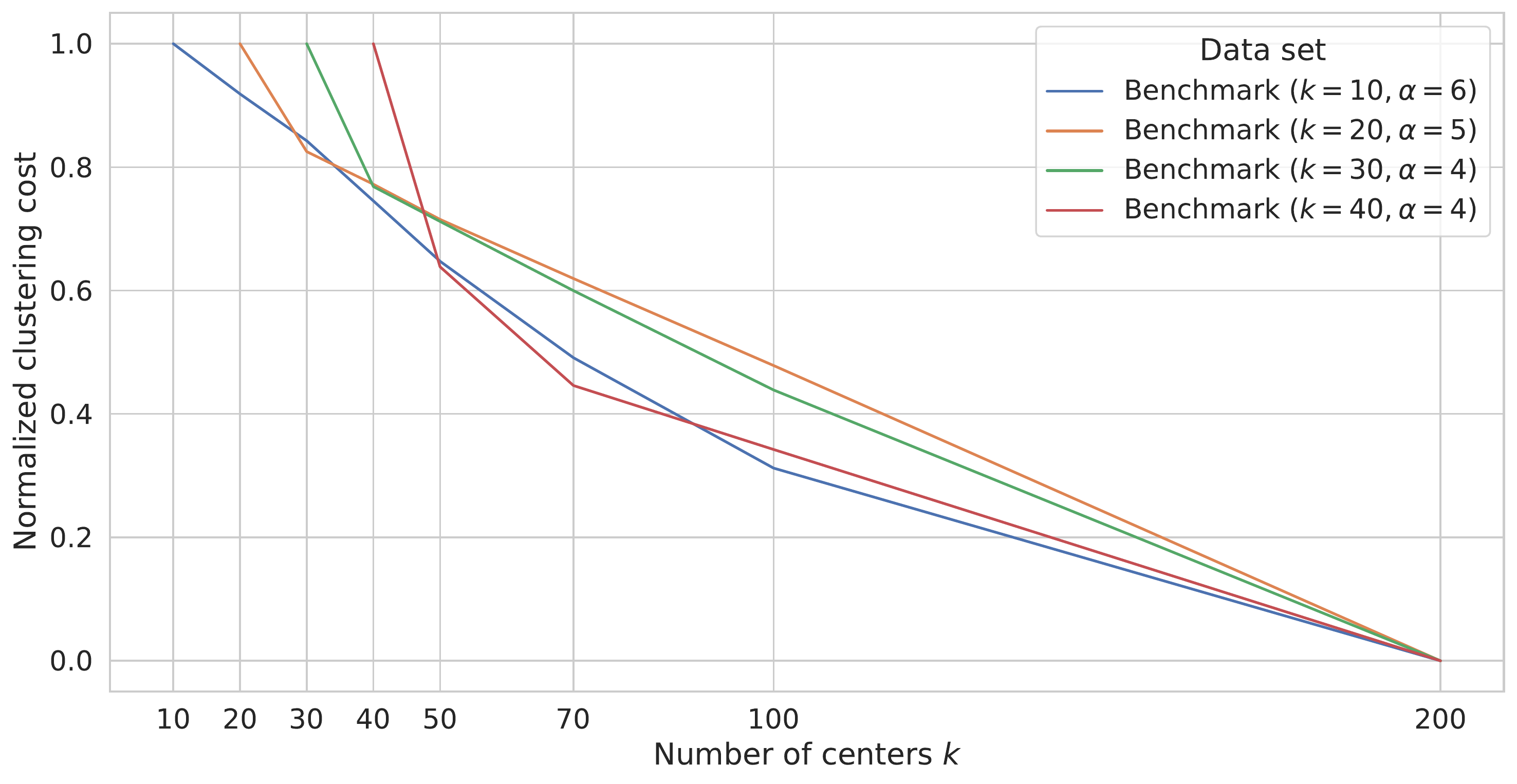}
  \caption{Shows the clustering costs of four instances of the benchmark framework as a function of the number of centers. In contrast to real-world data sets, the costs do not decrease rapidly as more cluster centers are added.
  }
  \label{fig:cost-curves-benchmark}
\end{figure}

\newpage

\begin{figure}[H]
 \includegraphics[width=.67\linewidth]{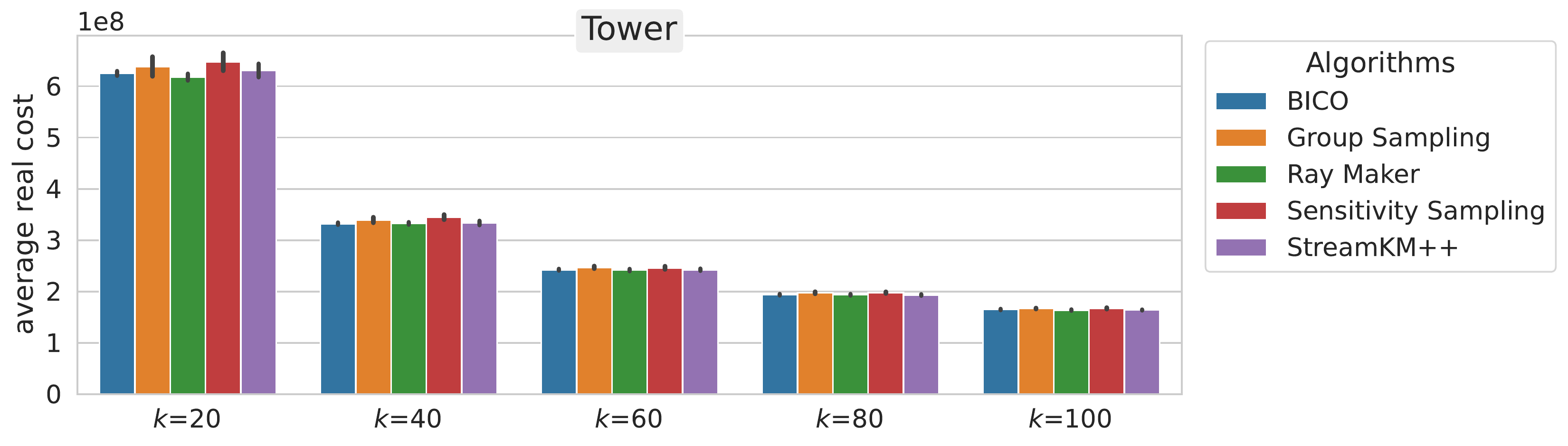}
 \newline
 \subfloat{
   \includegraphics[width=.47\textwidth]{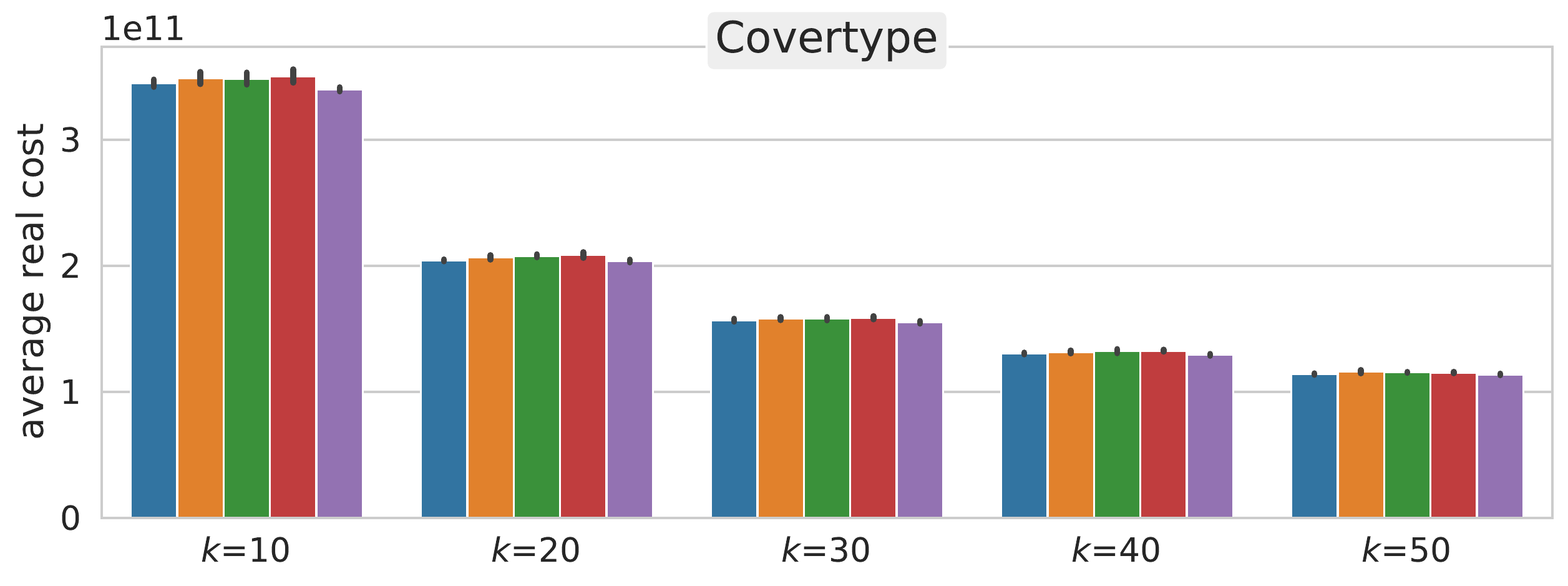}
 }
 \subfloat{
   \includegraphics[width=.47\linewidth]{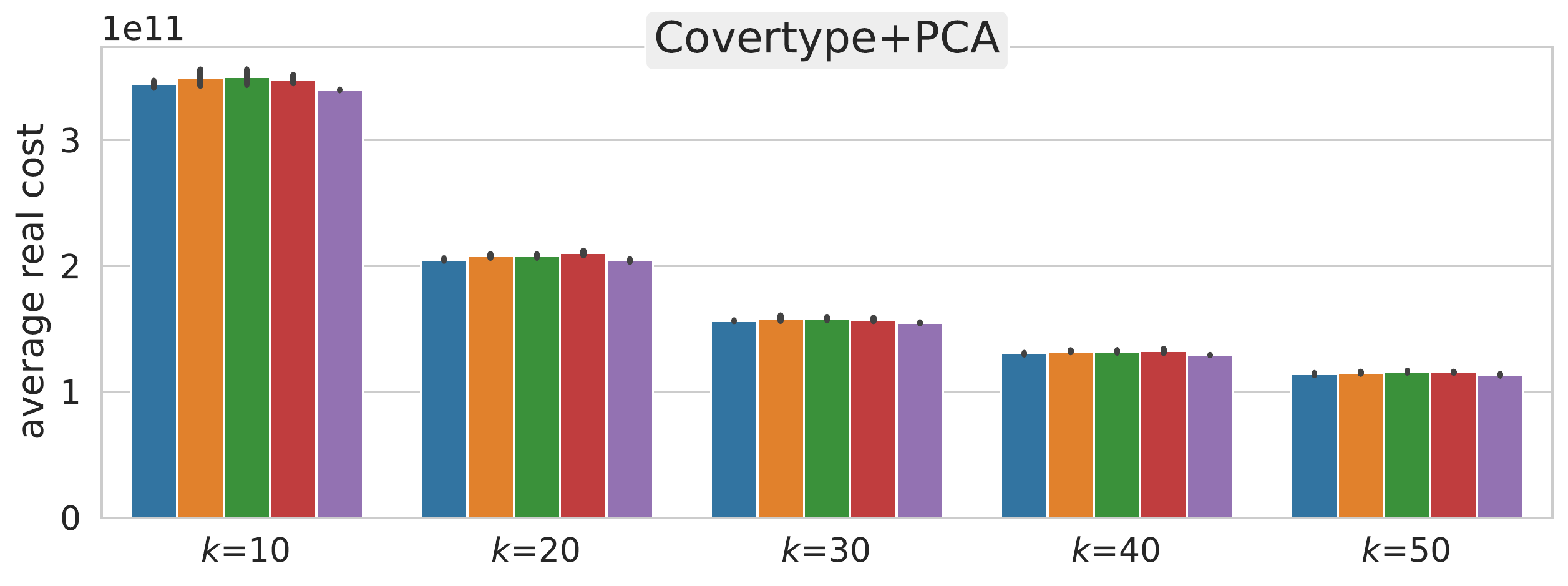}
 }
 \newline\newline
 \subfloat{
   \includegraphics[width=.47\textwidth]{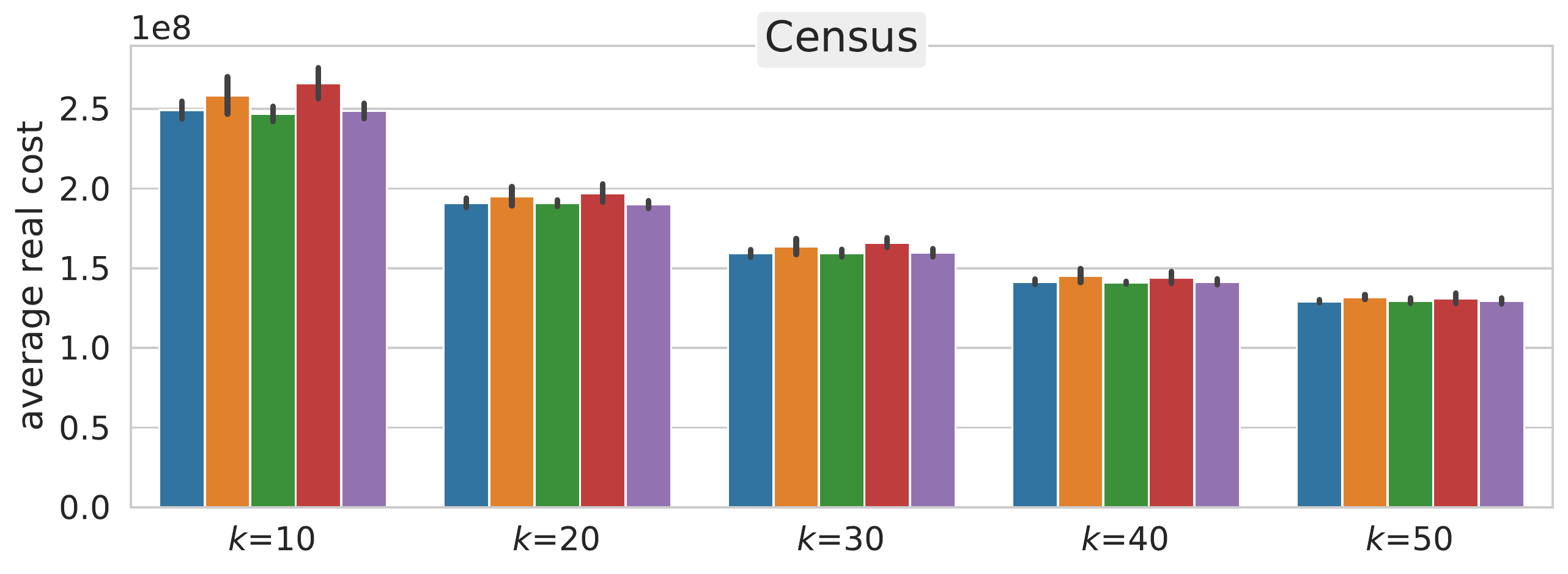}
 }
 \subfloat{
   \includegraphics[width=.47\linewidth]{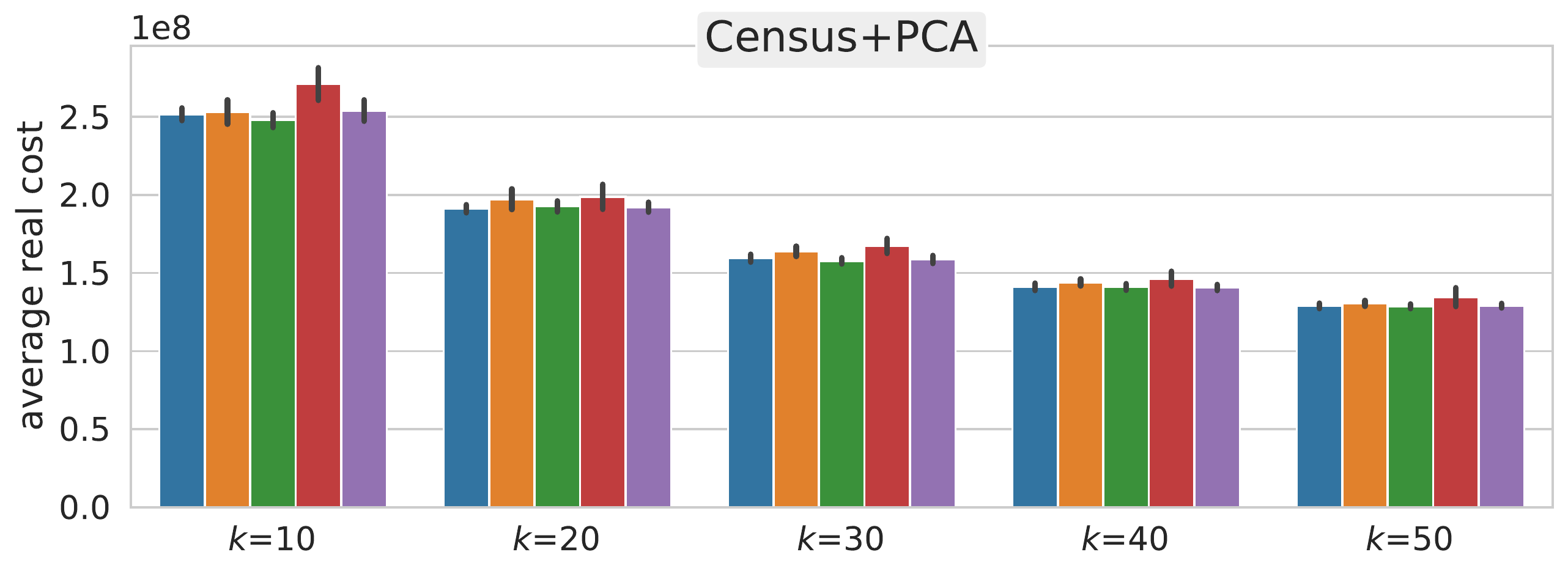}
 }
 \newline\newline
 \subfloat{
   \includegraphics[width=.47\textwidth]{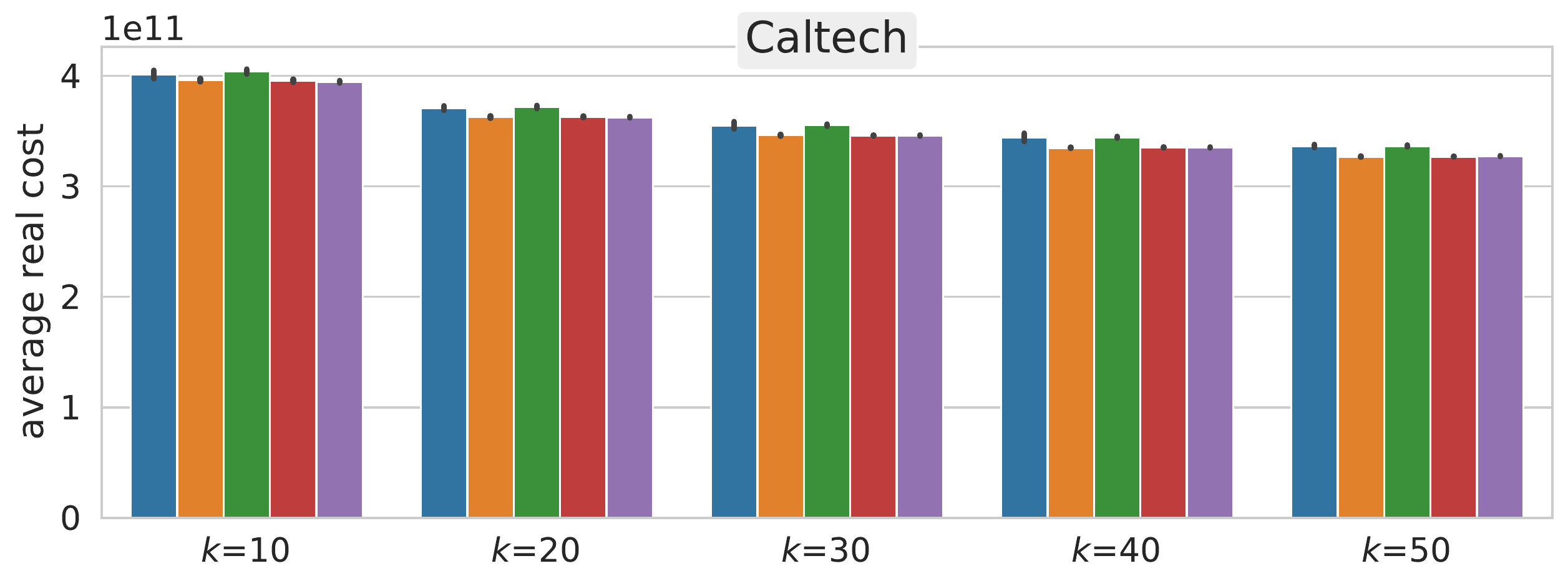}
 }
 \subfloat{
   \includegraphics[width=.47\linewidth]{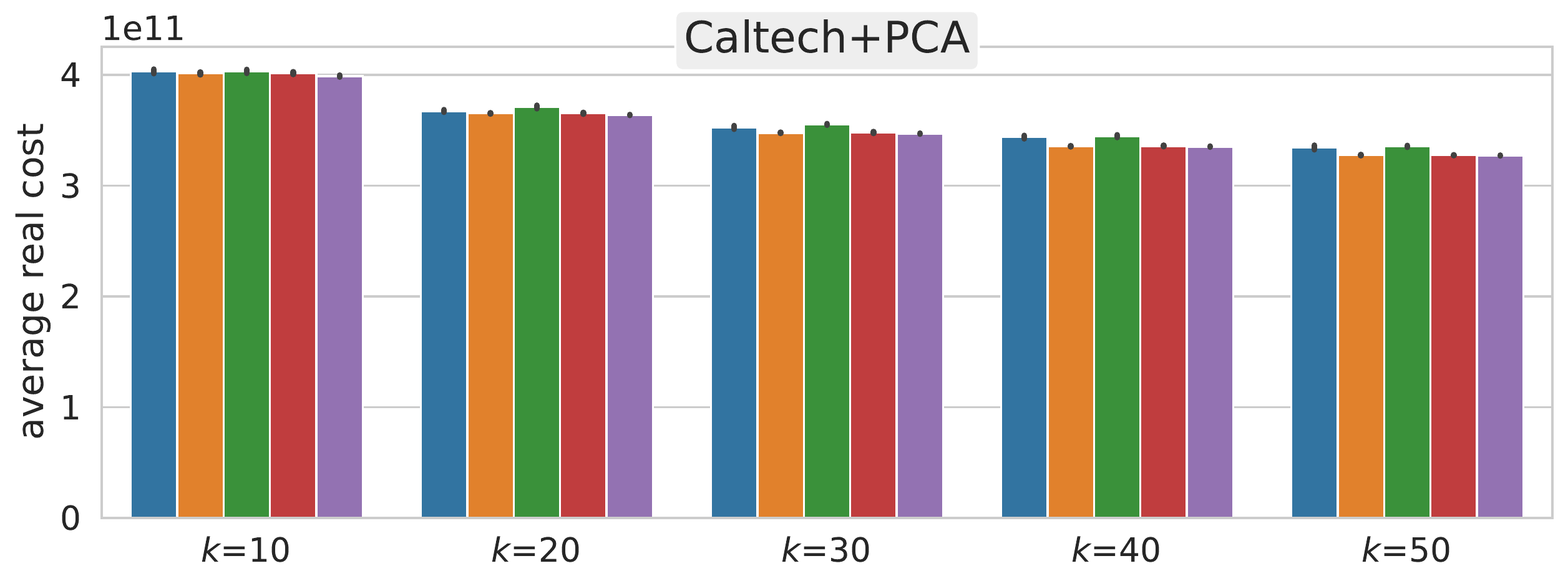}
 }
 \newline\newline
 \subfloat{
   \includegraphics[width=.47\textwidth]{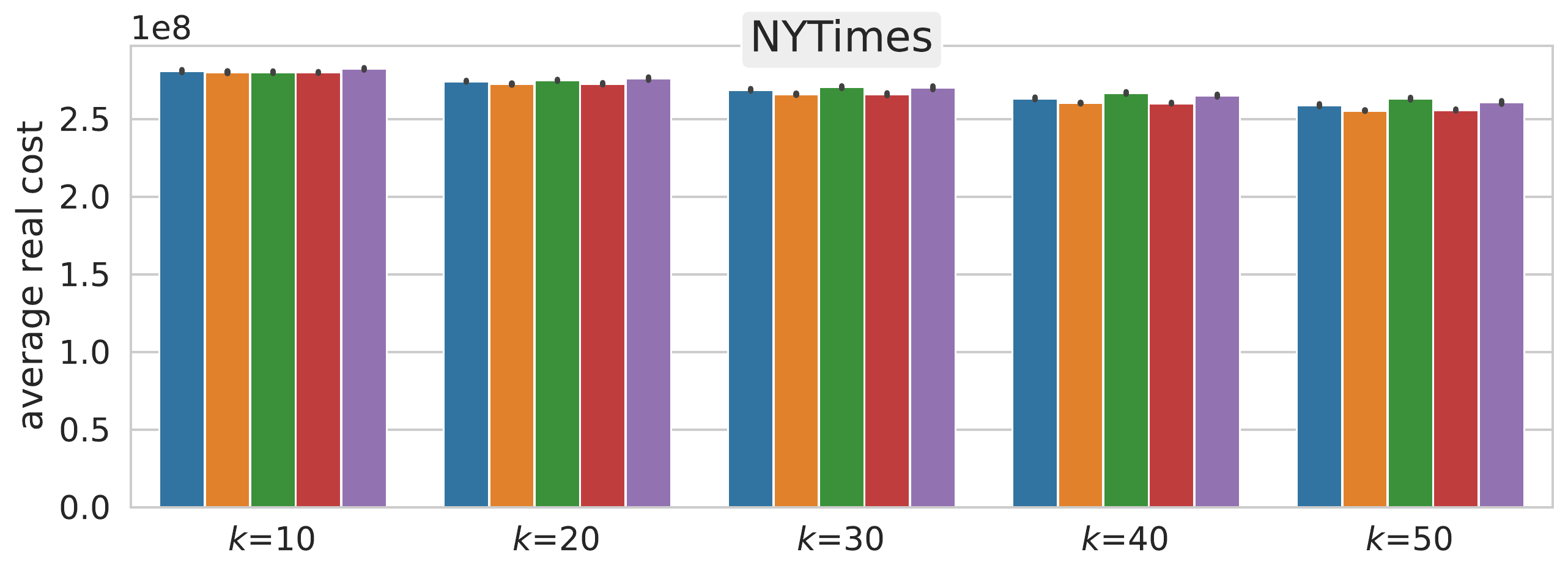}
 }
 \subfloat{
   \includegraphics[width=.47\linewidth]{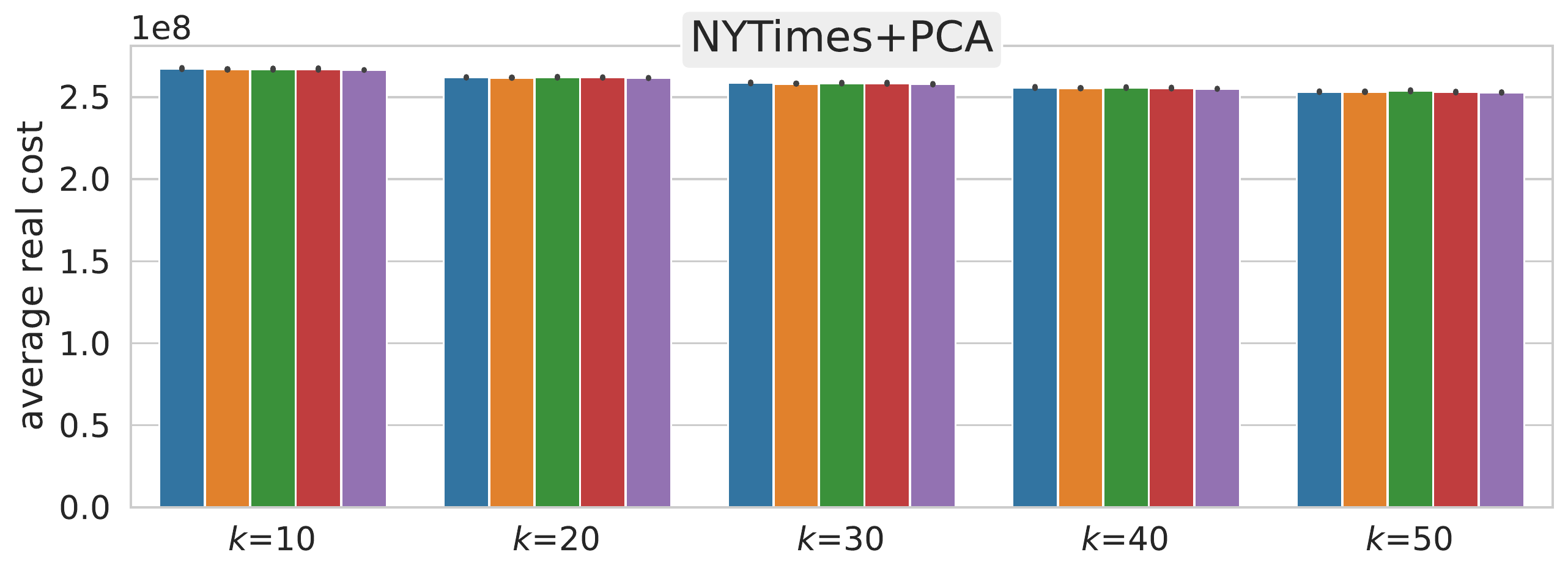}
 }
 \newline\newline
 \subfloat{
   \includegraphics[width=0.15\textwidth]{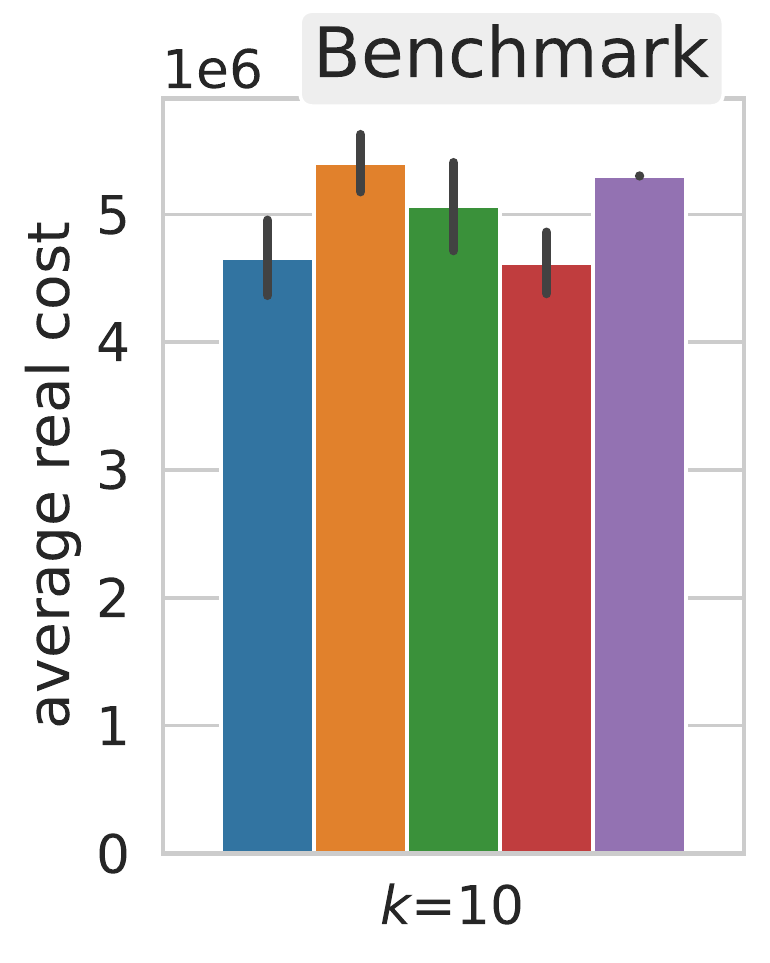}
   \includegraphics[width=0.165\textwidth]{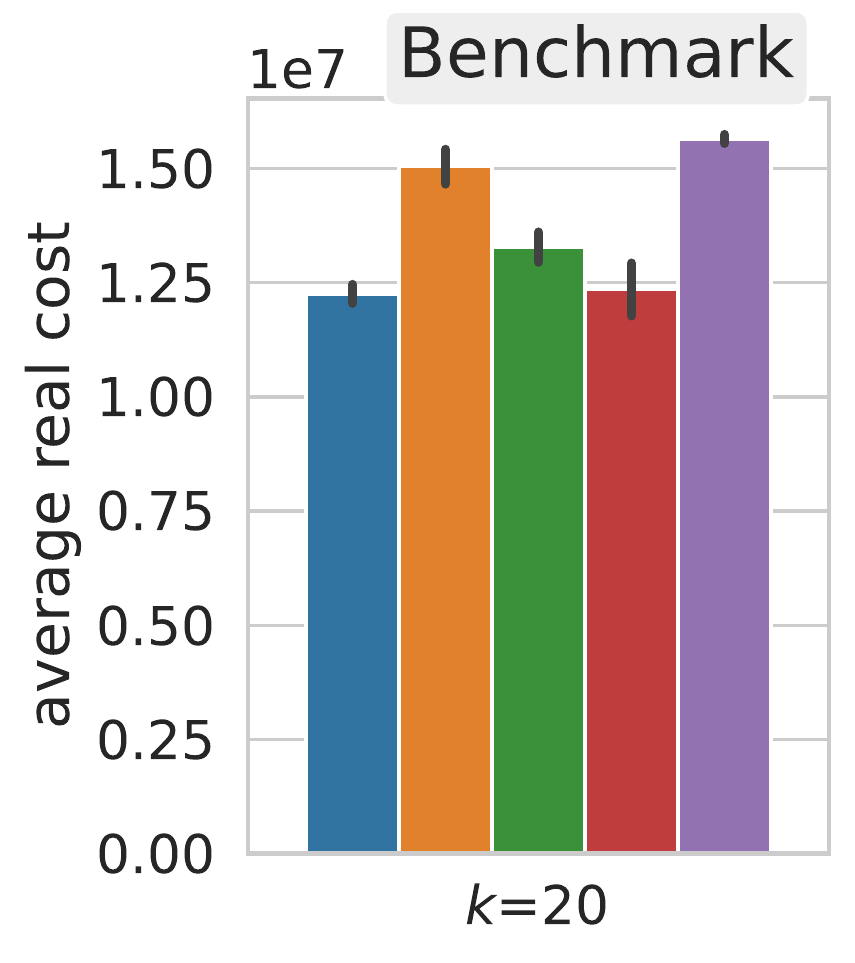}
   \includegraphics[width=0.16\textwidth]{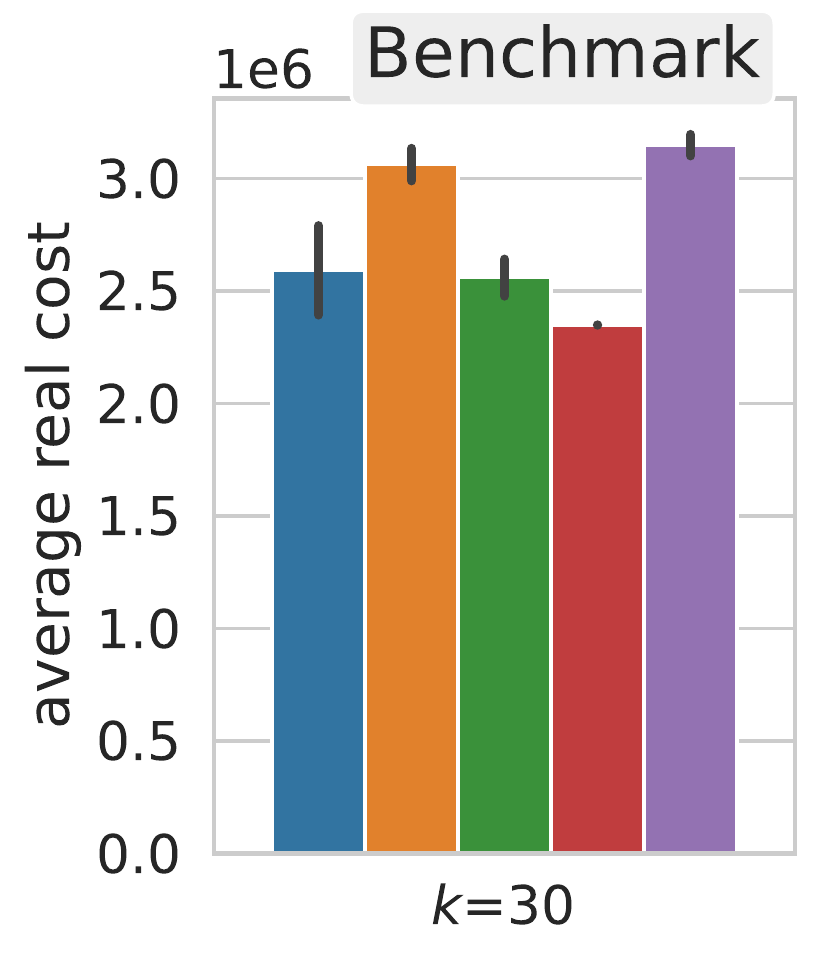}
   \includegraphics[width=0.31\textwidth]{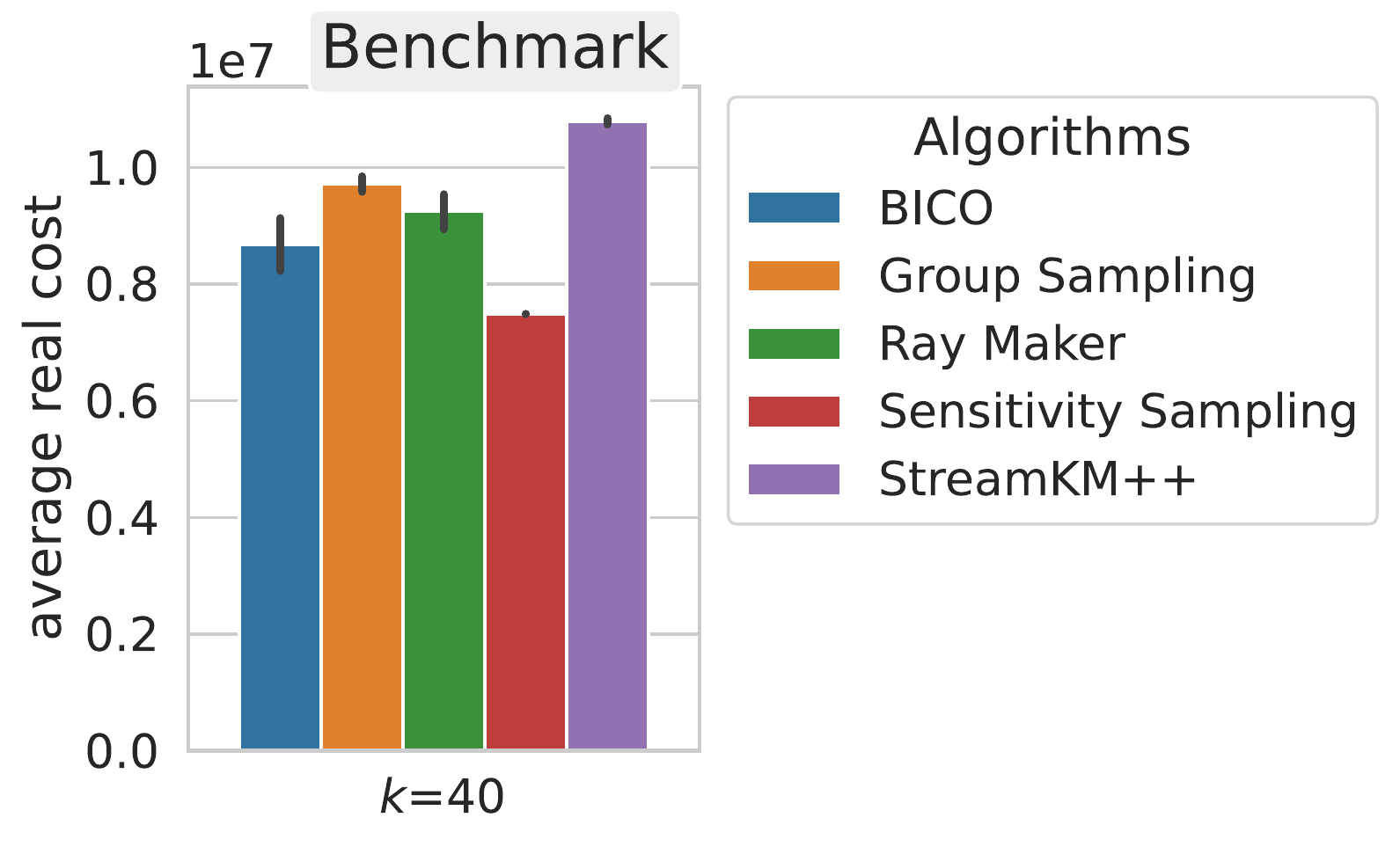}
 }
  \caption{The average costs (and standard deviations indicated by black bars) of running the evaluated coreset algorithms multiple times on different data sets. In general, the five coreset algorithms are able to compute coresets which result in solutions with comparable costs on the different real-world data sets. The differences in cost is more noticeable on the benchmark instances. Here, Senstivity Sampling is the winner because it seems to be better at capturing the correct ``clusters'' inherent in the benchmark instances.}
 \label{fig:real-costs}
\end{figure}



\end{document}